\documentclass[ecta,nameyear,draft]{econsocart} 

\usepackage{amsmath,blkarray,bm}
\usepackage{amsfonts} 
\usepackage{amsthm} 
\usepackage{amssymb}
\usepackage{array} %
\usepackage{bbm} 
\usepackage{booktabs}
\usepackage{caption}
\usepackage{centernot}
\usepackage{colortbl}
\usepackage{etoolbox}
\usepackage{enumitem}
\usepackage[mathscr]{euscript}
\usepackage{amssymb}
\usepackage{subcaption}
\usepackage{hhline}
\usepackage{pgfplots} 
\usepackage{pdfsync} 
\usepackage{longtable}
\usepackage{tabularx}
\usepackage{rotating}
\usepackage{threeparttable}
\usepackage{mathtools}
\usepackage{multirow}
\usepackage{setspace} 
\usepackage{scalerel}
\usepackage{stmaryrd}
\usetikzlibrary{shapes.geometric, arrows,positioning}
\usetikzlibrary{arrows}
\usepackage{tikz}
\usepackage{xr} 
\usepackage{xcolor}

\usetikzlibrary{positioning}

\usepackage{lscape}
\usepackage{arydshln}
\usepackage{amssymb}
\usepackage{algorithm}
\usepackage{algpseudocode}
\usepackage{pdflscape}
\makeatletter
\newcommand*{\addFileDependency}[1]{% argument=file name and extension
	\typeout{(#1)}% latexmk will find this if $recorder=0
	\@addtofilelist{#1}
	\IfFileExists{#1}{}{\typeout{No file #1.}}
}\makeatother

\newcommand*{\myexternaldocument}[1]{%
	\externaldocument{#1}%
	\addFileDependency{#1.tex}%
	\addFileDependency{#1.aux}%
}

\myexternaldocument{supplement}
\RequirePackage[colorlinks,citecolor=blue,linkcolor=blue,urlcolor=blue,pagebackref]{hyperref}

\startlocaldefs

\numberwithin{equation}{section}
\theoremstyle{plain}
\newtheorem{assumption}{Assumption} 
\newtheorem{theorem}{Theorem}
\newtheorem{lemma}{Lemma}
\newtheorem{proposition}{Proposition}[section] 
\newtheorem{corollary}{Corollary}[section]  

\theoremstyle{definition}
 
\numberwithin{assumption}{section}
\numberwithin{lemma}{section}
\numberwithin{definition}{section}

\theoremstyle{remark}

\def\argmax{\mathop{\rm arg\,max}}

\endlocaldefs

\begin{document}

	\renewcommand\thefootnote{\alph{footnote}}
	
	\begin{frontmatter} 
		\title{Penalized Likelihood for Dyadic Network Formation Models with Degree Heterogeneity}
		\runtitle{ } 
		\begin{aug}
			\author[id=au1,addressref={add1}]{\fnms{Zizhong}~\snm{Yan}\ead[label=e1]{helloyzz@gmail.com}}
			\author[id=au2,addressref={add2}]{\fnms{Jingrong}~\snm{Li}\ead[label=e2]{lijingrong@scau.edu.cn}}
			\author[id=au3,addressref={add1}]{\fnms{Yi}~\snm{Zhang}\ead[label=e3]{yzhang31@jnu.edu.cn}}
			\address[id=add1]{%
				\orgdiv{Institute for Economic and Social Research},
				\orgname{ Jinan University}}
			\address[id=add2]{%
				\orgdiv{College of Economics and Management},
				\orgname{South China Agricultural University}}
		\end{aug}
		
		\support{We are grateful to Mingli Chen, Qihui Chen, Iv\'{a}n Fern\'{a}ndez-Val, Bryan Graham, Lung-Fei Lee, Arthur Lewbel, Shuyang Sheng, Wei Shi, Ji-Liang Shiu, Xun Tang, Ao Wang, and Chao Yang for helpful comments and suggestions. We also thank seminar and conference participants at CUHK-Shenzhen, SHUFE, Warwick, Sun Yat-Sen, Jinan, and CIEP. 
			Yan acknowledges the support from National Natural Science Foundation of China (No. 72103079). 
			Li and Yan acknowledge the  Grant for Humanities and Social Sciences Research of the Ministry of Education (No. 25YJC91000 and 25JDSZ3124). 
			Zhang acknowledges  the National Natural Science Foundation of China (No. 72203077).
			All remaining errors are our own.
		}

		\begin{abstract}
			Estimating network formation models with degree heterogeneity raises two problems  in empirical networks. 
			First, agents that send no links, receive no links, or link to all remaining agents can make the fixed-effects MLE fail to exist. Trimming these agents changes the estimation sample and induces selection bias. 
			Second, the incidental-parameter problem biases common parameters and average partial effects. We resolve both issues through a \textit{penalized likelihood} approach.
			Our leading specification is a directed network model with reciprocity, nesting the standard undirected and non-reciprocal directed models.
			The penalty guarantees finite-sample existence and yields bias corrections for coefficients and partial effects. 
			We establish asymptotic results without imposing compactness on the fixed-effects.
			Allowing the  fixed effects to diverge at a logarithmic rate, our asymptotic framework accommodates the degree sparsity ubiquitous in large empirical networks.
			A global trade application demonstrates that our estimator avoids selection bias and recovers robust parameters where conventional methods fail.
		\end{abstract}
		\begin{keyword}
			\kwd{Dyadic network formation}
			\kwd{Degree heterogeneity}
			\kwd{Penalized likelihood}
			\kwd{Sparse networks}
			\kwd{Incidental parameter bias}
			\kwd{Average partial effects}
		\end{keyword}
	\end{frontmatter}
	
	%%%%%%%%%%%%%%%%%%%%%%%%%%%%%%%%%%%%%%%%%%%%%%%%%%%%%%%%%%%%%%%%%%%%%%%%%
	%%%% Main text entry area:
	%%%%%%%%%%%%%%%%%%%%%%%%%%%%%%%%%%%%%%%%%%%%%%%%%%%%%%%%%%%%%%%%%%%%%%%%%
	
	%\newpage
	\setcounter{footnote}{0}
	\renewcommand\thefootnote{\arabic{footnote}}

\section{Introduction}
Empirical analysis of network formation routinely relies on dyadic models to relate the presence of a link to observable characteristics of the pair.
Because link formation also depends on persistent agent specific attributes that are often unobserved by the researcher, it is standard practice to control for degree heterogeneity through sender and receiver fixed effects. In many empirical networks, however, this creates an estimation trilemma.

First, the fixed effects maximum likelihood estimator (MLE) frequently fails to exist in finite samples. Crucially, this existence failure is not confined to \textit{globally} sparse networks, it routinely occurs in moderately dense data whenever \textit{local} degree sparsity is present, such as agents sending no links, receiving no links, or connecting to everyone. 
In practice, researchers attempt to restore existence by trimming these boundary-degree nodes. However, in directed networks, this mechanical trimming is iterative and cascading. The resulting MLE is therefore computed on an endogenously selected subnetwork rather than the full observed network, inducing severe sample selection bias.
Second, even when the standard MLE exists, the \textit{incidental parameter problem} induces first order bias in the common parameters \citep{neyman1948consistent}. 
Third, while analytical bias corrections exist, extending them to structural functions such as average partial effects (APEs) typically requires cumbersome higher-order derivative calculations.

We address these  problems jointly with a \textit{penalized likelihood} (PL) approach to dyadic network formation models with degree heterogeneity. The PL estimator remains feasible in network samples with boundary-degree nodes, and yields bias corrected estimates of both coefficients and average structural functions.
Our leading specification is a directed network formation model with both unilateral incentives and bilateral reciprocity. Let $g_{ij}=1$ if agent $i$ forms a link to agent $j$, and let $g_{ij}=0$ otherwise. The payoff to agent $i$ takes the form
\begin{equation*}
	\mathcal{U}_{ij} = \underbrace{ g_{ij} \cdot X_{ij}'\beta}_{\text{directed utility}}  \quad  + \quad  \underbrace{ g_{ij} g_{ji} \cdot Z_{ij}'\rho }_{ \text{mutual utility} }  \quad  + \quad    \underbrace{g_{ij} \cdot ( \alpha_i + \gamma_j)}_{\text{degree heterogeneity}} ,
\end{equation*}
where $\alpha_i$ and $\gamma_j$ denote sender and receiver heterogeneity, and the mutual utility term captures reciprocity via bilateral dependence in link formation. Following the dynamic game approach of \citet{mele2017structural}, we show that, under an ergodicity condition, a stochastic best response dynamic yields a tractable stationary likelihood and a unique equilibrium distribution. This foundation motivates the likelihood, and our PL estimator is built on it.

The reciprocal directed specification serves as our leading case because it structurally nests the standard dyadic models in the literature. 
Setting the mutual utility parameter to zero ($\rho=0$) yields the directed model without reciprocity \citep{jochmans2018semiparametric, dzemski2019empirical, yan2019statistical, hughes2025estimating}, while equating the sender and receiver effects ($\alpha_i=\gamma_i$) under symmetric interactions recovers the undirected model \citep{graham2017econometric}. 
We develop the estimator and asymptotic theory for the general reciprocal case, and the same penalization logic applies to these benchmark models as special cases.

This paper makes four econometric contributions.
First, the PL approach guarantees the estimator exists in finite sample without trimming isolated or boundary-degree nodes. 
The penalty keeps the objective function well behaved when fixed effects drift toward the parameter boundary. 
This avoids the sample selection bias induced by repeated trimming the  network data.

	Second, the same likelihood-based penalty corrects the leading incidental parameter bias in the common parameters. The penalty is constructed from the block-diagonal part of the fixed-effects Hessian, with a \(2\times2\) sender-receiver Hessian block for each node in the reciprocal directed model.
	Our penalty is non-data-dependent and does not require the estimation of additional expectation terms. 
	By contrast, an estimator-level analytical correction can also be derived in the reciprocal setting. However, it still requires existence of the unpenalized estimator and therefore does not resolve the finite-sample nonexistence problem.
	Conditional likelihood methods provide useful benchmarks in nested settings. In the reciprocal directed model, however, they rely on six-node \textit{hexiad} cycles rather than tetrads or quadruples, which is computationally more demanding in large networks. In addition, because fixed effects are conditioned out, recovery of APEs is less direct.\footnote{For our reciprocal setting, we also develop a \textit{hexiad} conditional likelihood approach, which generalizes the \textit{tetrad} logit of \citet{graham2017econometric} and the \textit{quadruple} logit of \citet{jochmans2018semiparametric}. See Appendix \ref{sec:clogitappendix} for details.}

	Third, we extend the correction to APEs. This matters because empirical work typically reports partial effects, average link probabilities, and related structural functions rather than coefficients alone.
	After correcting the likelihood, we show the remaining leading bias in APEs comes from fixed effects uncertainty and can be removed with a simple adjustment based on the second derivatives of the APE function and the fixed-effects Hessian.

	Beyond the finite-sample existence and bias correction, we establish asymptotic existence, consistency, and asymptotic normality without standard compactness assumptions. In the existing literature, assuming a compact parameter space for fixed effects implicitly forces expected node degrees to grow proportionally with the network size, leading to a dense limiting network. We break this limitation by establishing asymptotic results while allowing the true fixed effects to diverge at a logarithmic rate. This growth condition accommodates the degree sparsity often observed in large empirical networks.
	
	Fourth, we approximate the inverse fixed-effects Hessian needed for sparse network asymptotics in the reciprocal directed model. 
	In benchmark undirected, non-reciprocal directed, and two-way fixed effects panel models, a diagonal approximation to the fixed-effects inverse Hessian is often sufficient \citep[see, e.g.,][]{fernandez2016individual,hughes2025estimating}. 
	Reciprocity breaks this simplification because sender and receiver effects are coupled within each node and across reciprocal dyads. 
	We derive a closed form hybrid approximation that combines \(2\times2\) node level Hessian blocks with a low-rank correction for two aggregate weak directions. The approximation remains valid when the network grows and some link probabilities shrink. It provides the uniform bounds used for consistency, asymptotic normality, and the APE correction without compactness assumptions on the fixed effects.
	
	We then apply the method to the global textiles trade network. This network is not globally sparse. Yet the standard fixed-effects MLE is still not available on the observed network. 11 countries have zero in-degree, and trimming them induces 5 major economies (France, Germany, Italy, Japan, and the UK) to become full out-degree countries in the remaining network. The application shows that PL estimation remains feasible on  the full observed network rather than on a selected subnetwork.\footnote{To facilitate empirical use, we provide a Python package, \texttt{NetworkFm}, that implements the PL estimator and the nested models used in the paper: \url{github.com/zizhongyan/NetworkFm}.}

	\textbf{Related literature.} This paper connects to the literatures on fixed-effects estimation in nonlinear panel and network formation models. In nonlinear panel models with two-way fixed effects, bias correction for parameters and average partial effects is extensively studied \citep[][among others]{fernandez2016individual, jochmans2019likelihood, chen2021nonlinear, leng2025debiased, yan2026robust}.\footnote{Because the PL correction acts on the likelihood itself, it preserves the extremum structure of the objective function. This makes likelihood-based inference a natural extension, in line with recent results for nonlinear panel models in \citet{leng2025debiased} and \citet{yan2026robust}. We do not develop that theory here.}
	In the network econometric literature, \citet{graham2017econometric}, \citet{jochmans2018semiparametric,dzemski2019empirical,hughes2025estimating,li2025bagging} develop fixed-effects estimators. From a statistical perspective, \citet{yan2016asymptotics,yan2019statistical} establish asymptotic properties for dyadic network models under diverging parameter spaces. \citet{ma2021determining} identify latent community structures via random block models rather than individual fixed effects. We build on this literature by simultaneously accommodating reciprocity, ensuring finite-sample existence without network trimming, and providing a unified likelihood-based correction for both coefficients and APEs. 
	
	A related literature studies network formation under broader strategic interactions and network externalities, including \cite{miyauchi2016structural, mele2017structural, boucher2017my, de2018identifying, leung2019inference, sheng2020structural, hoshino2022pairwise, ridder2025two, leung2026normal, menzel2026strategic}, and related work. We take a narrower focus on dyadic dependence generated by reciprocity, which yields a tractable estimator for sparse directed networks with unrestricted degree heterogeneity.

	\textbf{Outline of the paper.}
	Section \ref{sec:microUtility} introduces the reciprocal model and the likelihood foundation. 
	Section \ref{sec:MLE} presents the penalized likelihood estimator and the APE correction. 
	Section \ref{sec:asy} establishes the existence results and the large sample theory. 
	Section \ref{sec:montecarlo} reports the Monte Carlo evidence. 
	Section \ref{sec:empirical} presents the trade application. 
	The final section concludes. Proofs and additional computational details are collected in the appendices.

	\section{Dyadic Network Formation Model and Likelihood Foundation}\label{sec:microUtility}
	\subsection{ Payoff  Specification and Nested Models}
	Consider a network of $n$ agents (e.g., individuals, firms, or countries). The observed network is an $n\times n$ adjacency matrix $G=(g_{ij})$, where $g_{ij}\in\{0,1\}$ indicates whether agent $i$ forms a directed link to agent $j$. In the trade network, for example, a directed link indicates that one country exports to another. We rule out self-links, so $g_{ii}=0$ for all $i$.
	Let $\mathcal{U}_{ij}$ denote the payoff to agent $i$ from establishing a directed link to agent $j$. We propose a transferable utility function that incorporates both directed utility and mutual reciprocity:
	\begin{equation} 
		\mathcal{U}_{ij} = g_{ij}X_{ij}'\beta  + g_{ij}g_{ji}Z_{ij}'\rho + g_{ij}\left(\alpha_i + \gamma_j\right).
		\label{utility}
	\end{equation} 
	where the first directed utility term captures unilateral incentives to form a link from $i$ to $j$, while the second mutual utility term captures reciprocity through bilateral interdependence. $\alpha_i$ and $\gamma_j$ represent sender and receiver effects, respectively; $X_{ij}$ are covariates that enter the directed utility, and  $Z_{ij}$ are covariates that enter the mutual component.  
	The micro-foundation requires the covariates $Z_{ij}$ driving the reciprocal effect to be symmetric across dyads ($Z_{ij} = Z_{ji}$), whereas the directed covariates $X_{ij}$ can be asymmetric. We assume both $X_{ij}$ and $Z_{ij}$ include a constant term to capture \textit{baseline linking utilities}.

	The payoff specification in  \eqref{utility} provides a unified framework that encompasses two benchmark models that are standard in the literature:
	
	1) \textit{Directed networks without reciprocity.} If we assume that link formation relies purely on unilateral incentives, we shut down the mutual utility component by restricting $\rho = 0$. The payoff collapses to
	\begin{equation} 
		\mathcal{U}_{ij} = g_{ij}\left(X_{ij}'\beta + \alpha_i + \gamma_j\right),
		\label{utility2}
	\end{equation} 
	which yields the standard directed model with degree heterogeneity studied by \citet{jochmans2018semiparametric,yan2019statistical,dzemski2019empirical}, among others.
	Under this restriction, agent $i$'s decision to link with  $j$ is entirely independent of agent $j$'s linking decision.

	2) \textit{ Undirected networks.} 
	An undirected network requires that connections are symmetric by construction, so that $g_{ij} = g_{ji} \equiv g_{ij}^*$. 
	The same specification also arises as a special case of the directed model when link payoffs depend only on reciprocity.
	In this case, the dyadic outcome is whether two agents are mutually linked, which can be written as
	$g_{ij}g_{ji}\equiv g_{ij}^*$.
	Under symmetry, sender and receiver heterogeneity are no longer separately distinguished and collapse to a single node-specific effect, so that $\alpha_i=\gamma_i$.
	The payoff reduces to
	\begin{equation} 
		\mathcal{U}_{ij} = g_{ij}^* \left(Z_{ij}'\rho + \alpha_i + \alpha_j\right),
		\label{utility3}
	\end{equation} 
	which is the undirected specification studied by \citet{graham2017econometric}. 
	
	In Section  \ref{sec:MLE}, we develop the penalized likelihood  estimator for the generalized reciprocal model in  \eqref{utility} as the leading case, and our estimation and bias-correction methods automatically accommodates these benchmark cases.

	\subsection{ Stationary Equilibrium and Likelihood}
	
	In the directed model without reciprocity \eqref{utility2}, each link decision is unilateral, so the unique equilibrium is immediate. In the setting \eqref{utility3}, the dyad is summarized by a single symmetric link indicator, so the unique equilibrium is again straightforward. 
	
	However, once reciprocity is introduced to  the directed model \eqref{utility}, the pair $(g_{ij},g_{ji})$ becomes jointly determined. 
	The presence of the mutual utility term $g_{ij}g_{ji}Z_{ij}'\rho$ introduces strategic interdependence between agent $i$ and agent $j$. 
	As \citet{graham2020network} notes, combining both unilateral incentives and bilateral reciprocity in cross-sectional network data requires additional structure.  
	It is also well documented that in static, simultaneous-move games with such bilateral interactions, the network formation process generally suffers from multiple equilibria \citep{lewbel2007coherency, ciliberto2009market}. \citet{hoshino2022pairwise} highlights that such multiple equilibria render  an  incomplete model.
	
	To bypass this problem and obtain a well-defined likelihood, we follow the dynamic game approach of \citet{mele2017structural}. We model the observed cross-sectional network as the long-run realization of a stochastic best-response dynamic process \citep{blume1993statistical}. 
	Suppose that in each period $t$, agent $i$ meets $j$ and decides whether to update the connection $g_{ij}^t$ to maximize utility, taking $j$'s previous connection status $g_{ji}^{t-1}$ as given. Before updating, agent $i$ receives an idiosyncratic preference shock $\epsilon^t$. Agent $i$ sets $g_{ij}^t = 1$ if and only if $\mathcal{U}_{i}(g_{ij}^t = 1, g_{ji}^{t-1}, X_{ij}, Z_{ij}) + \epsilon_1^t \ge \mathcal{U}_{i}(g_{ij}^t = 0, g_{ji}^{t-1}, X_{ij}, Z_{ij}) + \epsilon_0^t$. We impose the following standard conditions on this dynamic process:
	\begin{assumption}[Dynamic network formation] \label{assumption1}
		The dynamic network formation game satisfies:
		\begin{enumerate}[label=(\roman*)]
			\item (Symmetry in mutual utility covariates) $Z_{ij} = Z_{ji}$ for all $i \neq j$.
			\item (Meeting technology) Every pair $(i, j)$ has a strictly positive probability of meeting in each period, independent of their existing links $g_{ij}$ and $g_{ji}$.
			\item (Idiosyncratic shocks) The preference shocks $\epsilon_0^t$ and $\epsilon_1^t$ are i.i.d.\ Type I extreme value distributed across all pairs and time periods.
		\end{enumerate}
	\end{assumption}
	Under Assumption \ref{assumption1}, the network formation process operates as a potential game \citep{monderer1996potential}. With strict positivity of meeting probabilities, the resulting Markov chain of network sequences is ergodic. Ergodicity ensures that the dynamic process converges to a unique stationary distribution, independent of the initial network configuration. This allows us to characterize the equilibrium distribution of the network as a standard exponential family model:
	\begin{proposition}[Equilibrium of the Dynamic Network Formation]\label{thm:micro}
		Suppose Assumption \ref{assumption1} holds. Then, the dynamic network formation game for any pair $(i,j)$ converges to a unique stationary distribution
		\begin{equation}
			\pi(g_{ij}, g_{ji} | \beta, \rho, \alpha_i, \gamma_j, \alpha_j, \gamma_i) = \frac{\exp(g_{ij}B_{ij} + g_{ji}B_{ji} + g_{ij}g_{ji}C_{ij})}{1 + \exp(B_{ij}) + \exp(B_{ji}) + \exp(B_{ij} + B_{ji} + C_{ij})}, \label{eqn:prop1pi}
		\end{equation}
		where $B_{ij} = X_{ij}'\beta + \alpha_i + \gamma_j$, $B_{ji} = X_{ji}'\beta + \alpha_j + \gamma_i$, and $C_{ij} = C_{ji} = Z_{ij}'\rho$.\footnote{
			\textit{Sketch of the proof.} The symmetry of $Z_{ij}$ enables the construction of a potential function that exactly maps the change in an individual's utility to a global change in the network state. Combined with the detailed balance condition, this ensures convergence to the unique Gibbs measure characterized by \eqref{eqn:prop1pi}. The formal proof is presented in  Appendix \ref{sec:appPropProof} of the Supplemental Material.
		}
	\end{proposition}
	Proposition \ref{thm:micro} provides the equilibrium distribution that is essential for our econometric strategy built in Section \ref{sec:MLE}. Because the model rules out external dependence beyond the dyad  $(i,j)$, the unique stationary distribution for the entire network of $n$ agents is simply the product of the dyadic distributions: $\prod_{i=1}^n \prod_{j>i}^n \pi(g_{ij}, g_{ji} | \cdot)$. 
	
	This dyadic independence yields a likelihood with a tractable denominator in \eqref{eqn:prop1pi} and keeps estimation within the frequentist MLE framework. By contrast, models with broader link interdependence typically lead to an intractable denominator and require Bayesian MCMC methods.\footnote{To overcome the estimation difficulty,  \cite{mele2017structural, mele2020does} applies the Bayesian double Metropolis-Hastings algorithm based on \cite{murray2012mcmc} to avoid the calculation of the denominator.  \cite{mele2023approximate} developed a variation inference algorithm to estimate a simplified version of the model.}

	\section{ Penalized Likelihood Estimators}\label{sec:MLE}
	This section details the estimation methodology. We first construct the log-likelihood function to show why standard MLE fails in sparse networks. We then introduce a PL estimator that restores finite-sample existence without data trimming and simultaneously corrects incidental-parameter bias.\footnote{In the context of MLE, the \textit{existence} of an estimator means that the supremum of the (log-)likelihood function can be achieved using a finite set of parameter values.} Finally, we provide a computationally tractable procedure to estimate bias-corrected APEs.
	
	\subsection{ Likelihood and  Existence}\label{sec:MLEexist}
	
	We adopt a fixed effects approach  for estimation, treating $\alpha$ and $\gamma$ as parameters to be jointly estimated  with  the common parameters $\theta = (\beta', \rho')'$. We collect  $2n-2$  fixed effects in the vector \textit{ordered by node}: 
	$\lambda=(\alpha_1,\gamma_1,\alpha_2,\gamma_2,\ldots,\alpha_{n-1},\gamma_{n-1})'$,
	subject to the location normalizations $\alpha_n = \gamma_n = 0$.\footnote{Equivalent normalizations are possible. Our choice keeps constant terms in both $X_{ij}$ and $Z_{ij}$, to capture baseline directed and reciprocal utility. Section \ref{sec:asy} returns to identification and  population maximization problem.}
	Using the stationary  distribution  in Proposition \ref{thm:micro}, the log-likelihood  for the observed network $G$ is
	\begin{align}  
		\ell(\theta,\lambda)=\textstyle\sum_{i=1}^n\sum_{j>i}^n \ell_{ij}\left(\theta,\alpha_i,\gamma_j,\alpha_j,\gamma_i \right),
		\label{eqn:loglikelihood1}
	\end{align}
	where
	$\ell_{ij}\left(\theta,\alpha_i,\gamma_j,\alpha_j,\gamma_i \right):=
	\ln\left[\pi(g_{ij},g_{ji}\mid \theta,\alpha_i,\gamma_j,\alpha_j,\gamma_i)\right]$. 
	
	It is useful to rewrite  \eqref{eqn:loglikelihood1} in exponential family form. Let the \textit{out-degree} and \textit{in-degree} of agent $i$ be $d_i:=\sum_{j\neq i}^n g_{ij}$ and $b_i:=\sum_{j\neq i}^n g_{ji}$, respectively. Define the corresponding \textit{degree sequence}s $d=(d_1,\ldots,d_{n-1})'$ and $b=(b_1,\ldots,b_{n-1})'$. The log-likelihood can then be expressed as
	$$
	\ell(\theta,\lambda)=\textstyle\sum_{i=1}^n\sum_{j\neq i}^n g_{ij}X_{ij}'\beta 
	+ \sum_{i=1}^n\sum_{j> i}^n g_{ij} g_{ji}Z_{ij}'\rho + d'\alpha + b'\gamma
	-\sum_{i=1}^n\sum_{j\neq i}^n \mathcal{C}_{ij}(\theta,\lambda),
	$$
	where $\mathcal{C}_{ij}(\theta,\lambda)$ is the log partition function.\footnote{The log partition function is
		$
		\mathcal{C}_{ij}(\theta,\lambda)=\frac{1}{2}\ln\big(1+\exp(B_{ij})+\exp(B_{ji})+\exp(B_{ij}+B_{ji}+C_{ij})\big),
		$
		where $B_{ij}$ and $C_{ij}$ defined in Proposition \ref{thm:micro}. Note that $\mathcal{C}_{ij}(\theta,\lambda)$ is symmetric in $i$ and $j$.
	}
	
	This exponential family representation explicitly isolates the sufficient statistics. For the common parameters $(\beta,\rho)$, they are $\sum_{i\neq j} g_{ij}X_{ij}$ and $\sum_{i<j} g_{ij}g_{ji}Z_{ij}$. For the fixed effects $\lambda$, the sufficient statistics are the out-degree and in-degree sequences $(d', b')'$.  This representation is convenient because the existence problem can be stated directly in terms of degree sequences.
	
	Let $(\widehat\theta, \widehat\lambda) = \argmax_{\theta,\lambda} \ell(\theta,\lambda)$ denote the fixed effects MLE, provided it exists. We assume that the common parameter space is compact and the covariates have bounded support. Under these standard regularity conditions, the finite-sample existence of the MLE dependent entirely on the fixed effects $\lambda$. More importantly, we do not impose compactness on the parameter space of $\lambda$, because doing so would rule out the slowly diverging heterogeneity that is central to the sparse network asymptotics studied later in the paper.
	
	For exponential family models, existence of the fixed effects MLE is equivalent to an interior condition on the sufficient statistics  \citep[see Theorem 5.5 of][]{brown1986fundamentals}. In the present model, this means that the observed bi-degree sequence $(d', b')$ must lie in the interior of the convex hull generated by all possible bi-degree sequences:
	\begin{align}  
		(d,b)\in
		\operatorname{int}\left(\operatorname{convhull}
		\left\{
		(d(G),b(G)):\;
		G\in\{0,1\}^{n\times n},\ g_{ii}=0
		\right\}\right). 
		\label{eqn:convexhull}
	\end{align}  
	
	Empirical  dataset routinely violate this convex-hull condition. The most obvious violations occur when an agent is an ``\textit{isolated node}''---sending no links ($d_i = 0$) or receiving no links ($b_i = 0$)---or a ``\textit{fully connected node}'' ($d_i = n-1$ or $b_i = n-1$). To circumvent this issue, researchers often trim such nodes until the remaining sample becomes estimable. In directed networks, however, this procedure may need to be repeated, because deleting one node can create new null in-degrees or out-degrees in the remaining graph. The standard MLE is then based on an endogenously selected subnetwork rather than on the observed one. 
	In the  trade application of Section \ref{sec:empirical}, for example,  MLE is only computable after trimming countries with zero in-degree. 
	This is precisely the problem that the penalized likelihood estimator is designed to avoid.

	Nevertheless, bounding the nodal degrees strictly away from $0$ and $n-1$ is merely a necessary condition for existence, not a sufficient one. As \citet{rinaldo2013maximum} demonstrate, there are still numerous other cases where the MLE is undefined, and these cases can be difficult to detect in practice.\footnote{\citet{rinaldo2013maximum} focus on the undirected $\beta$-model and systematically enumerate graph configurations on 4--6 nodes where MLE is undefined, despite all node degrees being strictly bounded away from zero and $n-1$.} 
	
	Even when the  MLE exists in a sufficiently dense network, it still suffers from the incidental parameter problem.  Because the dimension of $\lambda$ grows with the sample size $n$, estimation noise in $\widehat{\lambda}$ contaminates the concentrated likelihood, inducing a first-order bias of $\mathcal{O}(n^{-1})$ in $\widehat{\theta}$, which subsequently distorts structural functions such as APEs.
	
	Section \ref{sec:asy}  establishes the \textit{asymptotic} existence properties, and shows that the estimator can exist \textit{asymptotically} if the true fixed effects diverge sufficiently slowly.
	In finite samples, however, non-existence remains a severe barrier in empirical practice.
	The penalized likelihood estimator introduced next is designed specifically to resolve this finite-sample failure and the incidental parameter bias simultaneously, without   data trimming.

	\subsection{Penalized Likelihood and Bias Reduction}\label{sec:MLEPL}
	We construct the penalty term using the block-diagonal extraction of the negative fixed-effects Hessian. 
	Let ${D}(\theta,\lambda) = \operatorname{blkdiag}(D_1(\theta,\lambda), \dots, D_{n-1}(\theta,\lambda))$ denote this extraction from $-H_{\lambda\lambda}(\theta, \lambda) = - \partial_{\lambda\lambda'} \ell(\theta, \lambda)$. 
	For our leading reciprocal directed model, each block $D_i(\theta, \lambda)$ is a $2 \times 2$ matrix corresponding to the sender and receiver effects $(\alpha_i, \gamma_i)$ of node $i$. We define the penalty via the log-determinant:
	\begin{align}
		\eta(\theta, \lambda) = \tfrac{1}{2} \ln \det \bigl( {D}(\theta, \lambda) \bigr) = \tfrac{1}{2} \textstyle\sum_{i=1}^{n-1} \ln \det D_i(\theta, \lambda),
		\label{eqn:penaltymatrixform}
	\end{align}
	yielding the penalized likelihood estimator $(\widehat{\theta}_{\mathrm{PL}}, \widehat{\lambda}_{\mathrm{PL}}) = \argmax_{\theta,\lambda} \left[ \ell(\theta, \lambda) + \eta(\theta, \lambda) \right]$.
	
	A compelling theoretical advantage of this formulation is that it  adapts to the nested benchmark models discussed in Section \ref{sec:microUtility}. In a directed model without reciprocity ($\rho=0$), link formation depends purely on unilateral incentives, rendering the cross-partial derivatives $\partial_{\alpha_i \gamma_i}\ell$ identically zero. The block $D_i(\theta, \lambda)$ then naturally collapses to a strictly diagonal matrix, reducing the penalty to the standard sum of univariate log-Hessians: $\frac{1}{2} \sum_{i=1}^{n-1} [\ln(-\partial_{\alpha_i \alpha_i} \ell) + \ln(-\partial_{\gamma_i \gamma_i} \ell)]$. 
	
	Similarly, in an undirected network where degree heterogeneity is captured by a single node-specific parameter ($\alpha_i = \gamma_i$), the block $D_i$ degenerates to a scalar. The penalty then mathematically simplifies to the exact adjustment for one-way fixed effects: $\frac{1}{2} \sum_{i=1}^{n-1} \ln(-\partial_{\alpha_i \alpha_i} \ell)$. Equation \eqref{eqn:penaltymatrixform} thus provides a completely unified bias-reduction framework across the entire family of dyadic models.

	The mechanism of $\eta(\theta,\lambda)$ is twofold. First, it restores finite-sample existence. When any sender or receiver effect drifts to $\pm\infty$ (e.g., due to an isolated node),  the corresponding diagonal Hessian term collapses toward zero, so $\eta(\theta,\lambda)\to -\infty$. The penalized log-likelihood therefore cannot be maximized at the boundary. Theorem \ref{thm:lambdaBC_exist} in the next section  formalizes this argument. 
	Second, the same penalty removes the leading bias in the common parameters. Intuitively, the score contribution from $\eta(\theta,\lambda)$ offsets the leading bias in the score for $\theta$. Section \ref{sec:asy} makes this statement precise.  Moreover, the penalty is \textit{non-data-dependent}. It is determined entirely by the model-implied probabilities and can be evaluated directly.
	
	We compare the PL estimator with two alternative estimators for the reciprocal model:
	
	\textit{1) Relation to Estimator-level Correction (EC)}.
	For comparison, it is useful to note that the reciprocal model also admits an analytical estimator-level correction when the fixed effects MLE exists. Using the bias expansion derived in Section \ref{sec:asy}, one can form
	$$
	\widehat\theta_{\mathrm{EC}}
	=
	\widehat\theta 
	-
	\tfrac{1}{n-1}\widehat{\mathcal{I}}_n^{-1}\widehat{ \mathscr{B}}_n,
	$$
	where $\widehat{\mathcal{I}}_n$ and  $\widehat{ \mathscr{B}}_n$  are consistent estimates of the information matrix and the leading bias term in the score, respectively. However, $\widehat\theta_{\mathrm{EC}}$ is unavailable whenever the underlying fixed effects MLE does not exist. We therefore treat  $\widehat\theta_{\mathrm{EC}}$ as a benchmark estimator.

	\textit{2) Relation to Conditional Logit}.
	Conditional likelihood methods provide another benchmark. In the undirected case, \citet{graham2017econometric} established the tetrad logit method, while in the directed case without reciprocity, \citet{jochmans2018semiparametric} proposed a quadruple logit. In our reciprocal  case, a conditional likelihood construction can in principle be built from six-node cycles, which we discuss in Appendix \ref{sec:clogitappendix} of  Supplemental Material. For the reciprocity setting that motivates this paper, the PL approach is computationally simpler.

	\subsection{Average Partial Effects}\label{sec:APEintro}
	Empirical network research typically interested in structural functions, such as partial effects or counterfactual link probabilities, not only in coefficients. Let  
	$
	p_{ij}(\theta,\lambda)  
	=  
	\Pr(g_{ij}=1\mid X,Z,\theta,\lambda) ,
	$
	denote the model-implied directed link probability, and let $W_{ij,k}$ denote a generic dyadic regressor, which may be a component of either $X_{ij}$ or $Z_{ij}$, and let $\Delta_{ij}(\theta,\lambda)$ denote the corresponding dyad-specific effect on $p_{ij}(\theta,\lambda)$.
	
	For a continuous regressor $W_{ij,k}$, we define
	$
	\Delta_{ij}(\theta,\lambda)
	=
	\frac{\partial p_{ij}(\theta,\lambda)}{\partial W_{ij,k}}.
	$
	For a binary regressor $W_{ij,k}\in\{0,1\}$, we define
	$
	\Delta_{ij}(\theta,\lambda)
	=
	p_{ij}(\theta,\lambda; W_{ij,k}=1)
	-
	p_{ij}(\theta,\lambda; W_{ij,k}=0),
	$
	holding all other covariates fixed. When $W_{ij,k}$ is a component of $Z_{ij}$, the partial effect is understood as a dyad-level change applied symmetrically to $(i,j)$ and $(j,i)$, since $Z_{ij}=Z_{ji}$ in the reciprocal specification. 
	In the trade application, these partial effects summarize how bilateral frictions affect unilateral and mutual trade links.
	We define the APE by  
	$$
	\Delta(\theta,\lambda)  
	=  
	\tfrac{1}{n(n-1)}  
	\textstyle\sum_{i\neq j}^{n}  
	\Delta_{ij}(\theta,\lambda).  
	$$
	
	The plug-in APE estimator, $\Delta(\widehat\theta_{\mathrm{PL}},\widehat\lambda_{\mathrm{PL}})  =  \frac{1}{n(n-1)}  \sum_{i\neq j}  \Delta_{ij}(\widehat\theta_{\mathrm{PL}},\widehat\lambda_{\mathrm{PL}})$, suffers from asymptotic bias. As established in the two-way fixed effects nonlinear panel literature \citep{fernandez2016individual, yan2026robust}, the first-order bias in APEs stems from three sources: (i) bias in the common parameters; (ii) bias in the fixed-effects point estimates; and (iii) variance of the fixed-effects estimates.
	Evaluating the APE at the PL estimates thus automatically eliminates the first two sources of bias. To remove the remaining first-order bias, we define the following  bias-corrected APE estimator  
	\begin{align}
		\widetilde{\Delta}(\widehat\theta_{\mathrm{PL}},\widehat\lambda_{\mathrm{PL}})  
		=
		\Delta(\widehat\theta_{\mathrm{PL}},\widehat\lambda_{\mathrm{PL}})  
		-  
		\tfrac{1}{2}  
		\operatorname{tr} \left(  [
		\partial_{\lambda\lambda'}  
		\Delta(\widehat\theta_{\mathrm{PL}},\widehat\lambda_{\mathrm{PL}})]  
		\cdot
		[ S(\widehat\theta_{\mathrm{PL}},\widehat\lambda_{\mathrm{PL}})] 
		\right). \label{eq:ape_pl}
	\end{align}
	This correction uses only the second-order derivative of the APE function  and the fixed effects Hessian. No third-order  derivatives are required.  
	
	This is one of the practical advantages of the PL approach. Estimator-level APE corrections can also be constructed, including for the reciprocal model, but they require more higher order derivative calculations. By contrast, \eqref{eq:ape_pl} keeps coefficient correction and APE correction within the same likelihood-based framework. The same logic applies to the nested directed and undirected benchmark models. Section \ref{sec:asyAPE} establishes the asymptotic validity of \eqref{eq:ape_pl}, and Section \ref{sec:montecarlo} shows that it performs well in finite samples.

	\section{Asymptotic Theory for Penalized Likelihood Estimator}\label{sec:asy}
	
	\subsection{Notation and Assumptions}
	We first introduce notation used throughout the theoretical analysis. 
	
	For $a_1,a_2\in\{0,1\}$, let
	$p_{ij}^{(a_1a_2)}(\theta,\lambda):=\Pr(g_{ij}=a_1,g_{ji}=a_2\mid \theta,\lambda),$
	and define the marginal link probability
	$p_{ij}(\theta,\lambda):=\Pr(g_{ij}=1\mid \theta,\lambda)=p_{ij}^{(11)}(\theta,\lambda)+p_{ij}^{(10)}(\theta,\lambda).$
	When no confusion can arise, we suppress arguments evaluated at the true values $(\theta_0,\lambda_0)$, for example, we write 
	$\ell=\ell(\theta_0,\lambda_0),$
	and
	$p_{ij}=p_{ij}(\theta_0,\lambda_0).$
	
	Let $H_{\theta\theta}=\partial_{\theta\theta'}\ell$, $H_{\theta\lambda}=\partial_{\theta\lambda'}\ell$ and  $H_{\lambda\lambda}=\partial_{\lambda\lambda'}\ell$, 
	where $\partial_x f$ denotes the derivative of $f$ with respect to $x$, and additional subscripts indicate higher-order derivatives. We write the Hessian of the log-likelihood as
	$$
	H=
	\begin{pmatrix}
		H_{\theta\theta} & \quad H_{\theta\lambda}\\
		H_{\lambda\theta} & \quad  H_{\lambda\lambda}
	\end{pmatrix},
	\qquad
	H_{\theta\theta}=
	\begin{pmatrix}
		H_{\beta\beta} & \quad H_{\beta\rho}\\
		H_{\rho\beta} &\quad  H_{\rho\rho}
	\end{pmatrix},
	\qquad
	H_{\lambda\lambda}=
	\begin{pmatrix}
		H_{\alpha\alpha} &\quad  H_{\alpha\gamma}\\
		H_{\gamma\alpha} & \quad H_{\gamma\gamma}
	\end{pmatrix}.
	$$
	
	Let
	$\theta_0=(\beta_0',\rho_0')'$, 
	$\lambda_0=(\alpha_{0,1},\gamma_{0,1},\alpha_{0,2},\gamma_{0,2},\ldots,\alpha_{0,n-1},\gamma_{0,n-1})'$
	denote the true parameter vector under the identifying restriction  $\alpha_n=\gamma_n=0$. 
	An equivalent normalization is
	$\sum_{i=1}^n \alpha_i=\sum_{i=1}^n \gamma_i.$
	We adopt $\alpha_n=\gamma_n=0$ because it allows both $X_{ij}$ and $Z_{ij}$ to retain a constant term. In the present model, the constant in $X_{ij}$ has a natural interpretation as baseline directed utility, while the constant in $Z_{ij}$ captures baseline reciprocal utility. 
	Under this identification, the true parameter values are defined as the solution to the population maximization problem
	$$(\theta_0,\lambda_0)=\argmax\limits_{\theta\in\Theta,\lambda\in\Lambda_n}\mathbb E_0\bigl[\ell(\theta,\lambda)\bigr],
	$$
	where $\mathbb E_0$ denotes expectation under the true distribution of the data conditional on the exogenous covariates.
	
	Throughout, $\Theta$ is compact, while $\Lambda_n\subset\mathbb{R}^{2n-2}$ is closed but need not be bounded and hence not compact. 
	The support of the covariates is bounded. 
	We do not impose compactness on $\Lambda_n$, the asymptotic analysis later allows $\|\lambda_0\|_\infty$ to grow at a logarithmic rate. This accommodates sparse networks in which the number of observed links per agent does not grow proportionally with the network size.\footnote{A compact parameter space for the fixed effects would effectively force nodal degrees to grow proportionally with network size in the limit, thereby inducing a dense network asymptotically. This restriction  can be strong for many empirical networks of interest \cite[see for example][]{graham2024sparse}.}
	We collect the baseline conditions below:

\begin{assumption}[Baseline econometric assumptions]
	\label{assumption2}
	Given the log-likelihood function in equation (\ref{eqn:loglikelihood1}), suppose the following conditions hold:
	\begin{enumerate}[label=(\roman*)]
		\item For each unordered dyad $(i,j)$, the conditional distribution of $(g_{ij},g_{ji})$ is characterized by the stationary dyadic likelihood in Proposition \ref{thm:micro}. The econometrician observes $(g_{ij},g_{ji},X_{ij}',X_{ji}',Z_{ij}')$, where $Z_{ij}=Z_{ji}$.
		\item The common parameter $\theta_0 = (\beta_0', \rho_0')'$ lies in the interior of a compact space $\Theta \subset \mathbb{R}^{\dim\theta}$, where $\dim\theta$ is fixed. For each $n$, the fixed-effects parameter space $\Lambda_n \subset \mathbb{R}^{2n-2}$ is closed.
		\item The covariates satisfy $X_{ij} \in \mathbb{X}$ and $Z_{ij} \in \mathbb{Z}$ for all $i \neq j$, where $\mathbb{X}$ and $\mathbb{Z}$ are compact subsets of $\mathbb{R}^{\dim\beta}$ and $\mathbb{R}^{\dim\rho}$, respectively.
		\item The population expected log-likelihood $\mathbb E_0[\ell(\theta,\lambda)]$ is uniquely maximized at $(\theta_0,\lambda_0)\in \Theta\times\Lambda_n$ for all sufficiently large $n$.
	\end{enumerate}
\end{assumption} 
Assumption \ref{assumption2} is standard for nonlinear estimation problems \citep{newey1994large}. The only point that differs from many fixed-effects formulations is that $\Lambda_n$ is assumed closed, but not compact.  Condition (iv) is a higher level condition on the uniqueness of $\theta_0$ and $\lambda_0$, provided that they exist.\footnote{Since the model belongs to the exponential family, Assumption \ref{assumption2}(iv) can equivalently be formulated as a rank condition on the sufficient statistic for large $n$. By Corollary 6.16 in \citet{lehmann1983theory}, full rank of the sufficient statistic implies that the model is minimal, and guarantees uniqueness of the parameter vector.}

	\subsection{Existence}
	Two different existence questions arise in this paper. The first is finite-sample existence. In a realized network, the fixed effects MLE may fail to exist because the observed bi-degree sequence lies on the boundary of the convex support, as discussed in condition \eqref{eqn:convexhull} of Section \ref{sec:MLEexist}. The second is asymptotic existence. Even when an estimator is well defined in a given sample, its large sample behavior still depends on how quickly the true fixed effects are allowed to grow with the network size. The penalty resolves the first problem directly. The second remains a large-sample issue. 
	We discuss finite-sample existence first.
	\begin{theorem}   
		\label{thm:lambdaBC_exist}
		Suppose Assumption \ref{assumption2} holds. Then, for every observed network, the penalized log-likelihood
		$
		\ell_{\mathrm{PL}}(\theta,\lambda)
		=
		\ell(\theta,\lambda)+\eta(\theta,\lambda)
		$
		attains its maximum on $\Theta\times\Lambda_n$. In particular,
		$
		(\widehat\theta_{\mathrm{PL}},\widehat\lambda_{\mathrm{PL}})
		\in
		\argmax_{\theta\in\Theta,\lambda\in\Lambda_n}
		\ell_{\mathrm{PL}}(\theta,\lambda)
		$
		exists.  
	\end{theorem}
	Theorem \ref{thm:lambdaBC_exist} is a finite sample result. It does not require the convex-hull condition \eqref{eqn:convexhull}, and therefore does not require trimming isolated or fully connected nodes before estimation. This is the practical existence property that distinguishes the penalized estimator from  estimator-level corrections that are defined only when the underlying MLE exists.
	
	Finite-sample existence alone does not imply asymptotic existence. 
	We still need to specify how the true fixed effects behave as the network grows. Following \cite{yan2016asymptotics,yan2019statistical}, we allow the maximum entry of the true fixed effects $\|\lambda_0\|_\infty$ diverges at a rate proportional to $\tau \ln n$ for a sufficiently small constant $\tau>0$.  While a compact parameter space for the fixed effects implies a dense limiting network, this logarithmic growth condition relaxes this restriction. 
	It does not force the whole limiting network to be sparse, but it allows  link probabilities to shrink as $n$ increases.\footnote{To see the intuition, consider the link probability $p_{ij} \approx e^{\alpha_i + \gamma_j} / [1 + e^{\alpha_i + \gamma_j}]$. Under the compactness assumption,  $p_{ij}$ is strictly bounded away from zero to one. Expected nodal degrees such as $\sum_{j \ne i} p_{ij}$ grow at rate $\mathcal{O}(n)$ and generate a dense limiting network. Under $\|\lambda_0\|_\infty \le \tau \ln n$, the link probabilities can instead decay at a polynomial rate $p_{ij}\approx e^{-2\tau \ln n} = n^{-2\tau}$ See \citet{graham2024sparse} for an in-depth discussions on the asymptotically sparsity.}
	
	In the global trade application studied in Section~\ref{sec:empirical}, expanding the network sample  $n$ typically involves including smaller or more remote economies. Consequently,  the probability that two countries form a trade link may become small. A standard compactness assumption would  \textit{force} two countries to trade when sample increases.  
	Our asymptotic framework accommodates this empirical concern.

	\subsection{Approximation of the Inverse Hessian}
	
	A central technical challenge in this paper is deriving an explicit approximation for the inverse fixed-effects Hessian, $-H_{\lambda\lambda}^{-1}$. This task is substantially more demanding in the reciprocal directed model than in standard nested models. 
	
	In the undirected model, the directed model without reciprocity, and standard two-way fixed-effects panels, the fixed-effects Hessian exhibits a strictly diagonal or near-block-diagonal structure. Consequently, an approximation based on its diagonal is often adequate: $-H_{\lambda\lambda}^{-1}\approx -\mathrm{diag}(H_{\lambda\lambda})^{-1}$ \citep[see, e.g.,][]{fernandez2016individual, hughes2025estimating}.\footnote{
		Related graph-based fixed effects results show that precision depends on connectivity and on the small eigenvalues of the graph Laplacian \citep[see][]{chung1997spectral,jochmans2019fixed}.  
	} Once reciprocity is introduced, however, this simplification breaks down. Because an agent's outgoing and incoming links are mutually dependent, the off-diagonal entries within the node-specific blocks ($H_{\alpha_i\gamma_i}$) and the cross-node blocks ($H_{\alpha_i\gamma_j}$) are strictly non-zero and scale as $\mathcal{O}(n)$. As a result, the standard diagonal approximation no longer applies.
	
	To bypass the computationally prohibitive high-dimensional inversion, we derive a closed-form matrix $S$ that analytically approximates the inverse Hessian. For each $i=1,\ldots,n-1$, we capture the node-specific information via the $2 \times 2$ matrix
	\begin{equation}
		D_i:=\left(
		\begin{smallmatrix}
			\sum_{j\neq i}^{n-1} p_{ij}(1-p_{ij})
			&  \quad
			\sum_{j\neq i}^{n-1} (p_{ij}^{(11)}-p_{ij}p_{ji})
			\\
			\sum_{j\neq i}^{n-1} (p_{ij}^{(11)} - p_{ij}p_{ji})
			& \quad
			\sum_{j\neq i}^{n-1} p_{ji}(1-p_{ji})
		\end{smallmatrix}\right),\label{eqn:Di}
	\end{equation}
	and define the block diagonal matrix $D:=\operatorname{blkdiag}\bigl(D_1,\ldots,D_{n-1}\bigr)$. To account for global network movements, we define the aggregate basis vectors $u_+:=(1,\ldots,1)'$ and $u_-:=(1,-1,\ldots,1,-1)'$, stacked as $U:=[u_+,u_-]$. We then define the approximate inverse as
	\begin{equation}
		S(\theta,\lambda(\theta)):=  D^{-1}(\theta,\lambda(\theta)) + U\bigl[U'[-H_{\lambda\lambda}(\theta,\lambda(\theta))]U\bigr]^{-1}U'.
		\label{eq:def_S_hybrid_simple}
	\end{equation}
	
	The approximation \(S\) is explicit and easy to compute. The first term, $D^{-1}$, requires only the inversion of $n-1$ independent $2\times 2$ blocks. In the trade network  application, it accommodates the local dependence between outward and inward trade propensities. The second term is a low-rank aggregate correction requiring the inverse of only a single $2 \times 2$ matrix, $U'[-H_{\lambda\lambda}]U$. This decomposition provides a computationally tractable closed-form substitute for the full $(2n-2)\times(2n-2)$ matrix inversion.
	
	To formalize the approximation precision of $S$ under sparse network asymptotics, we define a standardized Hessian  $\widetilde V:= D^{-1/2} [-H_{\lambda\lambda}] D^{-1/2}$. In this standardized space, the aggregate subspace is spanned by $\widetilde U:= D^{1/2}U$. Let $P:= \widetilde U\bigl(\widetilde U'\widetilde U\bigr)^{-1}\widetilde U'$ denote the orthogonal projector onto the column space of $\widetilde U$, and let $P^\perp:=I-P$ be the residual projector. 
	
	\begin{assumption}[Projected inverse regularity]
		\label{assumption:spectral_hybrid_simple}
		Suppose Assumption \ref{assumption2} holds. For all sufficiently large $n$, uniformly over $\theta\in\Theta$ and $\lambda\in \Lambda$,  $\widetilde V(\theta,\lambda(\theta))$ satisfies:
		\begin{enumerate}[label=(\roman*)] 
			\item (Global weak directions) The column space of $\widetilde U(\theta,\lambda(\theta))$ is the eigenspace of $\widetilde V(\theta)$ associated with its two smallest eigenvalues.
			\item (Strongly identified orthogonal subspace) The inverse of $\widetilde V(\theta,\lambda(\theta))$ restricted to the orthogonal subspace satisfies
			$
			\bigl\| P^\perp(\theta,\lambda(\theta))\widetilde V(\theta,\lambda(\theta))^{-1}P^\perp(\theta,\lambda(\theta))-P^\perp(\theta,\lambda(\theta)) \bigr\|= \mathcal{O}\big(n^{-1} \exp(4\|\lambda(\theta)\|_\infty) \big).
			$
		\end{enumerate}
	\end{assumption} 
	
	Assumption \ref{assumption:spectral_hybrid_simple} formally describes the algebraic structure of the standardized Hessian. Condition (i) isolates the two aggregate directions $u_{+}$ and $u_{-}$.  
	In a global trade network, $u_+$ represents a network-wide parallel shift in trade density, while $u_-$ raises  outward trade propensities and lowers  inward trade propensities by the same amount, or vice versa.  Since these global movements are weakly informative for identifying relative country differences, they are absorbed by the low-rank projection $U$. 
	
	Condition (ii) imposes regularity on the remaining directions. After the two aggregate movements are removed, the inverse standardized Hessian is close to the identity on the orthogonal subspace. 
	In the global trade example, countries differ fundamentally in their numbers of export and import partners. \(D_i\) measures the  information available for country \(i\)'s exporter and importer effects. Countries with more trading partners contribute more variations, while countries close to isolation contribute less. The condition allows such degree heterogeneity. 
	After this degree adjustment, condition (ii) requires that  the likelihood still has enough local variation to distinguish countries' relative outward and inward trade propensities.
	
	The rate in condition (ii) is compatible with Assumption \ref{assumption2}. Under the logarithmic growth condition used below, \(\|\lambda_0\|_\infty\le \tau\log n\), the right-hand side is of order $n^{-1+4\tau}$. Hence the condition does not require the fixed effects parameter space  to be compact.

	The next lemma gives the rate bounds for the inverse-Hessian approximation used below. 
	We state the result for a generic vector \(\mathbf f\) because the same calculation is applied to the fixed effects score, Hessian-score products, and higher order remainders. 
	The lemma uses two restrictions on \(\mathbf f\). The largest coordinate of \(\mathbf f\) and the aggregate sender-receiver difference are of the same order. 
	This formulation avoids applying a matrix norm to \(-H_{\lambda\lambda}^{-1}-S\), which would give a bound too crude for the subsequent asymptotic expansions. 
	In Lemma \ref{prop:hybrid_ranktwo_simple}, the first bound controls the approximation error, and the second controls the bound of \(S\mathbf f\).
	
	\begin{lemma}
		\label{prop:hybrid_ranktwo_simple}
		Suppose Assumptions \ref{assumption2} and \ref{assumption:spectral_hybrid_simple} hold. 
		Assume $\|\lambda_0\|_{\infty} = \tau\ln n$ with $\tau>0$.
		Then, for any vector $\mathbf f\in\mathbb R^{2n-2}$ satisfying
		$\|\mathbf f\|_\infty=\mathcal O_P(na_n)$ and 
		$|u_-'\mathbf f| = \mathcal O_P(na_n)$ 
		for some nonnegative random sequence $a_n$, one has\footnote{
			Because the fixed-effects vector is ordered by node as
			$
			\lambda=(\alpha_1,\gamma_1,\ldots,\alpha_{n-1},\gamma_{n-1})',
			$
			the projected difference associated with
			$
			u_-=(1,-1,\ldots,1,-1)'
			$
			is
			$
			u_-'\mathbf f = \sum_{i=1}^{n-1}(\mathbf f_{2i-1}-\mathbf f_{2i}),
			$
			which corresponds to grouping all sender effects first and all receiver effects second.
		}
		\begin{enumerate}[label=(\roman*)]
			\item 
			$
			\big\|[-H_{\lambda\lambda}^{-1}(\theta,\lambda(\theta))-S(\theta,\lambda(\theta))] \, \mathbf{f} \big\|_{\infty} 
			= \mathcal O_P\bigl( a_n \, n^{6\tau-1/2} \,   
			e^{6\|\lambda(\theta)-\bar{\lambda}(\theta)\|_\infty
				+
				6\|\bar\lambda(\theta)-\lambda_0\|_\infty  }\bigr),
			$
			\item 
			$
			\big\|S(\theta,\lambda(\theta)) \, \mathbf{f}\big\|_{\infty} 
			= \mathcal O_P\bigl( a_n \, n^{2\tau} \,   
			e^{2\|\lambda(\theta)-\bar{\lambda}(\theta)\|_\infty
				+
				2\|\bar\lambda(\theta)-\lambda_0\|_\infty  }\bigr),
			$
		\end{enumerate}
		uniformly over $\theta\in\Theta$ and $\lambda \in \Lambda$.
	\end{lemma}

	\subsection{Large-Sample Properties}\label{sec:asyresults}
	We now turn to the large-sample properties of the estimator. For each $\theta$, let
	$$
	\widehat\lambda(\theta)=\argmax_{\lambda\in\Lambda_n}\ell(\theta,\lambda),
	\qquad
	\bar\lambda(\theta)=\argmax_{\lambda\in\Lambda_n}\mathbb E_0[\ell(\theta,\lambda)].
	$$
	The common parameter estimator can then be viewed as:
	$
	\widehat\theta
	=
	\argmax_{\theta\in\Theta}
	\ell\bigl(\theta,\widehat\lambda(\theta)\bigr).
	$
	The incidental parameter problem arises because each component of $\widehat\lambda(\theta)$ is estimated from only $n-1$ observations, so the noise in $\widehat\lambda(\theta)$ contaminates the concentrated score for $\theta$. 
	Hence the limit distribution of $\widehat{\theta}$ does not center at $\theta_0$, unless the quantity $\widehat{\lambda}(\theta)$ was replaced by its \textit{infeasible} population quantity $\bar{\lambda}(\theta)$.
	This is the same mechanism that generates first-order bias in nonlinear panel models with two-way fixed effects.

	The first step is to establish asymptotic existence and consistency. The proof combines a Newton iteration for $\widehat\lambda(\theta)$ with the approximation in Lemma \ref{prop:hybrid_ranktwo_simple}. The log growth condition on $\|\lambda_0\|_\infty$ guarantees that the iteration remains in a neighborhood where the Hessian approximation is valid.
	\begin{theorem}\label{thm:lamba_consist}
		Suppose Assumptions \ref{assumption2} and \ref{assumption:spectral_hybrid_simple} hold.  Assume $\|\lambda_0\|_{\infty} \leq \tau \ln n$ with $0<\tau<1/8$. Then, with probability approaching one,
		\begin{enumerate}[label=(\roman*)]
			\item $\widehat{\lambda}(\theta_0)$ exists and satisfies
			$\|\widehat{\lambda}(\theta_0)-\lambda_0\|_{\infty}\leq\mathcal{O}_P\big({\sqrt{\ln(n)}n^{2\tau-1/2}}\big)=o_P(1).$
			\item
			$\lVert \widehat{\theta}-\theta_0 \rVert =o_P(1)$.
			\item $\widehat{\lambda}$ exists and satisfies $\|\widehat{\lambda} - \lambda_0 \|_{\infty}  = 
			\mathcal{O}_P\big({\sqrt{\ln(n)}n^{2\tau-1/2}}\big)
			+
			\mathcal{O}_P\big(n^{2\tau}\big)\|\widehat{\theta}-\theta_0\|.$
		\end{enumerate}
	\end{theorem}

	The next result gives the asymptotic distribution.
	\begin{theorem}\label{thm:lambda_asynormal}
		Suppose Assumptions \ref{assumption2} and \ref{assumption:spectral_hybrid_simple} hold. Assume $\|\lambda_0\|_{\infty} \leq \tau \ln (n)$ with  $0<\tau<1/12$. Let $N=n(n-1)$.
		Then,
		\begin{enumerate}[label=(\roman*)]
			\item $\sqrt{N}(\widehat{\theta}-\theta_0) 
			\stackrel{d}{\rightarrow} \mathcal{N}(
			\mathcal{I}^{-1}_{\infty} \mathscr{B}_{\infty}
			,\mathcal{I}_{\infty}^{-1}),
			$
			where
			$\mathcal{I}_{\infty}=-\lim_{n\rightarrow\infty}\big[\frac{1}{(n-1)^2}(H_{\theta\theta}+H_{\theta\lambda}SH_{\lambda\theta})\big]$, and
			$\mathscr{B}_{\infty}$ is the limit of $ \partial_{\theta}\eta $ as $n\to\infty$.
			\item For any fixed length $L \geq 1$, 
			$
			\sqrt{n}[\widehat{\lambda}-\lambda_0]_{1:L} \stackrel{d}{\rightarrow} \mathcal{N}\big(0, [\Omega_{\infty}]_{1: L, 1: L}\big),
			$
			where 
			$\Omega_{\infty}$ is the limit of $nS$ as $n\to\infty$.
		\end{enumerate}
	\end{theorem}
	Theorem \ref{thm:lambda_asynormal}(i) makes the incidental parameter problem explicit. The common parameter estimator is asymptotically normal, but centered at a nonzero bias term of order $1/n$. This is the asymptotic benchmark for the \textit{estimator-level correction (EC)} discussed in Section \ref{sec:MLEPL}.
	Next result  shows that the penalized likelihood correction ($\widehat{\theta}_{\mathrm{PL}}$) recenters the asymptotic distribution without inflating the variance.
	\begin{corollary}\label{thm:thetaBC_asynormal}
		Suppose Assumptions \ref{assumption2} and \ref{assumption:spectral_hybrid_simple} hold. Assume $\|\lambda_0\|_{\infty} \leq \tau \ln (n)$ with $0<\tau<1/12$. Let $N=n(n-1)$. Then, 
		$
		\sqrt{N}(\widehat{\theta}_{\mathrm{PL}}-\theta_0) 
		\stackrel{d}{\rightarrow} \mathcal{N}\left(
		0
		,\mathcal{I}_{\infty}^{-1}\right).
		$
	\end{corollary}

	\subsection{Bias Correction for Average Partial Effects}\label{sec:asyAPE}
	
	We next study average partial effects. The same argument also applies to generic \textit{Average Structural Functions (ASF)} that are smooth averages of dyad level link probabilities. 
	We state a high-level condition that is sufficient for the PL-based bias correction mechanism.  
	\begin{assumption}[Average structural functions]
		\label{assumption3}
		Let
		$
		\Delta(\theta,\lambda)=\frac{1}{n(n-1)}\sum_{i\neq j}\Delta_{ij}(\theta,\lambda),
		$
		Suppose the dyad-specific function $\Delta_{ij}(\theta,\lambda)$ is four times continuously differentiable in a neighborhood of
		$(\theta_0,\lambda_0)$. In that neighborhood, the functions
		$\Delta_{ij}(\theta,\lambda)$ and their derivatives with respect to the parameters up to the third order are uniformly bounded.
		Furthermore, the asymptotic variance of the linearized ASF, $V_{\Delta,\infty}$ (defined in Corollary \ref{cor:ape_pl}), is strictly positive and finite. 
	\end{assumption}
	
	Assumption \ref{assumption3} imposes standard smoothness and boundedness on the generic dyadic structural functions. 
	For the APEs studied in this paper, the derivative bounds follow from a simple count.   The function $\Delta(\theta,\lambda)$ averages over $n(n-1)$ directed dyads. 
	A common parameter enters every term in this average.  A fixed effect for node $i$ enters only the $2(n-1)$ directed dyads involving that node. 
	Thus, uniformly in a neighborhood of $(\theta_0,\lambda_0)$, we can verify that
	$
	\|\partial_\theta\Delta\|=\mathcal{O}(1),
	$
	$
	\|\partial_\lambda\Delta\|=\mathcal{O}(n^{-1/2}),
	$
	$
	\|\partial_{\lambda\lambda'}\Delta\|=\mathcal{O}(n^{-1}),
	$ 
	$
	\|\partial_{\theta\theta'}\Delta\|=\mathcal{O}(1),
	$
	and
	$
	\|\partial_{\theta\lambda'}\Delta\|=\mathcal{O}(n^{-1/2}).
	$\footnote{
		These rates follow from a simple count by the APE construction.
		For all $i$, we have $\max_i|\partial_{\lambda_i}\Delta | = \mathcal{O}(n^{-1})$. Hence, $\|\partial_\lambda\Delta\|\leq n^{1/2} \max_i|\partial_{\lambda_i}\Delta | =O(n^{-1/2})$.
		Moreover, because $\|\partial_{\lambda_i\lambda_i} \Delta \| = \mathcal{O}(n^{-1})$ and $\|\partial_{\lambda_i\lambda_j} \Delta \| = \mathcal{O}(n^{-2})$ for all $i\neq j$, it follows that $\|\partial_{\lambda\lambda'} \Delta\|_{\infty} = \|\partial_{\lambda\lambda'} \Delta\|_{1} =\mathcal{O}(n^{-1})$. This gives that $\|\partial_{\lambda\lambda'} \Delta\| \leq \sqrt{\|\partial_{\lambda\lambda'} \Delta\|_{\infty}\|\partial_{\lambda\lambda'} \Delta\|_{1}} = \mathcal{O}(n^{-1}).$ Finally, each entry of $\partial_{\theta\lambda'}\Delta$ is $O(n^{-1})$. Since $\dim(\theta)$ is fixed, $\|\partial_{\theta\lambda'}\Delta\| \leq n^{1/2}O(n^{-1})=O(n^{-1/2})$.
	}
	Taking more derivatives with respect to $\lambda$ does not enlarge the set of affected dyads. 
	%A third-order fixed effects remainder is therefore bounded by the cube of the fixed effects perturbation.
	These structural limits guarantee that the Taylor remainders of the structural function shrink at appropriate rates.

	Bias correction for coefficients does not automatically imply bias correction for APEs. As explained in Section \ref{sec:APEintro}, the PL estimator removes the first-order bias coming from the common parameters and the fixed-effects point estimates. The only remaining first-order term is induced by fixed-effects uncertainty. Equation \eqref{eq:ape_pl} subtracts exactly this bias, and the next corollary establishes the resulting asymptotic distribution.
	
	\begin{corollary}
		\label{cor:ape_pl}
		Suppose Assumptions \ref{assumption2}, \ref{assumption:spectral_hybrid_simple} and \ref{assumption3} hold. Assume $\|\lambda_0\|_\infty \le \tau \ln n$ with $0<\tau<1/12$. Then, 
		$\sqrt{N}\Bigl( \Delta(\widehat\theta_{\mathrm{PL}},\widehat\lambda_{\mathrm{PL}}) - \Delta- \tfrac{1}{2}\operatorname{tr}\bigl( [\partial_{\lambda\lambda'}\Delta] S \bigr) \Bigr) \overset{d}{\rightarrow} \mathcal N(0, V_{\Delta,\infty}),$
		where the asymptotic variance $V_{\Delta,\infty} :=\lim_{n \to \infty} V_{\Delta,n}$,
		with $V_{\Delta,n}$  defined in  \eqref{eq:APEvariance}.
	\end{corollary}

	Corollary \ref{cor:ape_pl} clarifies why the penalized likelihood approach is especially convenient for structural functions. After penalizing the likelihood, the only remaining bias term can be removed by subtracting the trace term above. It does not require third-order derivatives. In this sense, the PL approach keeps coefficient correction and APE correction within a single likelihood-based framework. 
	From Corollary \ref{cor:ape_pl}, the bias-corrected APE estimator  $\widetilde{\Delta}(\widehat\theta_{\mathrm{PL}},\widehat\lambda_{\mathrm{PL}})$ defined in equation \eqref{eq:ape_pl} 
	is asymptotically centered at the true APE.  
	Consequently,
	$
	\sqrt{N}\bigl(\widetilde{\Delta}(\widehat\theta_{\mathrm{PL}},\widehat\lambda_{\mathrm{PL}})-\Delta\bigr)\overset{d}{\rightarrow}\mathcal N(0,V_{\Delta,\infty}).
	$

	\section{Monte Carlo Experiments}\label{sec:montecarlo}
	This section studies the finite sample performance of the PL estimator in designs tailored to the existence problem. 
	We find that in sparse networks, where isolated nodes occur frequently, standard fixed effects methods often become unavailable, while the PL estimator remains computable and continues to deliver accurate bias correction. We benchmark the PL estimator against the uncorrected  MLE, the estimator-level correction (EC), and conditional logit methods. We evaluate  both the common parameters and the APEs.
	
	\subsection{Designs}\label{sec:DGP}
	Our data-generating processes adapt the simulation framework in \citet{graham2017econometric} to accommodate asymmetric links and reciprocity. We consider three model classes: (\textit{i}) the reciprocal directed model that is the leading case of the paper, (\textit{ii}) the directed model without reciprocity, and (\textit{iii}) the undirected model.

	\begin{table}[H]
		\caption{Monte Carlo Designs and Network Sparsity}\label{tab: Monte Carlo Designs}
		\centering
		{\fontsize{8.5pt}{9.0pt}\selectfont
			\begin{tabularx}{\linewidth}{
					>{\hspace*{1em}\raggedright\arraybackslash}X
					>{\centering\arraybackslash}m{0.055\linewidth}
					>{\centering\arraybackslash}m{0.085\linewidth}
					>{\centering\arraybackslash}m{0.085\linewidth}
					>{\centering\arraybackslash}m{0.025\linewidth}
					>{\centering\arraybackslash}m{0.085\linewidth}
					>{\centering\arraybackslash}m{0.085\linewidth}
					>{\centering\arraybackslash}m{0.085\linewidth}
				}
				\toprule
				& \multicolumn{3}{c}{\begin{tabular}{c}
						Symmetric\\
						Uncorrelated Heterogeneity
				\end{tabular}}
				& & \multicolumn{3}{c}{\begin{tabular}{c}
						Right-Skewed\\
						Correlated Heterogeneity
				\end{tabular}}\\
				\cline{2-4}\cline{6-8}
				\rule{0pt}{2.4ex} & A.1 & A.2 & A.3 & & B.1 & B.2 & B.3 \\
				\hline
				
				\multicolumn{8}{l}{\cellcolor{gray!7}\textit{Panel A. Parameter settings}}\\
				$\varrho_L$ & $-1/2$ & $-1$ & $-2$ & & $-2/3$ & $-7/6$ & $-13/6$ \\
				$\varrho_H$ & $-1/2$ & $-1$ & $-2$ & & $-1/6$ & $-2/3$ & $-5/3$ \\
				$\varpi_0$  & $1$    & $1$  & $1$  && $1/4$  & $1/4$  & $1/4$ \\
				$\varpi_1$  & $1$    & $1$  & $1$  && $3/4$  & $3/4$  & $3/4$ \\
				\hline
				
				\multicolumn{8}{l}{\cellcolor{gray!7}\textit{Panel B. Directed network model with mutual utility} $(n=100)$}\\
				Density           & $0.416$ & $0.219$ & $0.039$ && $0.451$ & $0.254$ & $0.049$ \\
				Transitivity      & $0.430$ & $0.233$ & $0.053$ && $0.487$ & $0.297$ & $0.071$ \\
				\% Succ.\ of MLE  & $1.000$ & $1.000$ & $0.002$ && $1.000$ & $1.000$ & $0.003$ \\
				\hline
				
				\multicolumn{8}{l}{\cellcolor{gray!7}\textit{Panel C. Directed network model without mutual utility} $(n=100)$}\\
				Density           & $0.315$ & $0.166$ & $0.032$ && $0.344$ & $0.193$ & $0.040$ \\
				Transitivity      & $0.397$ & $0.239$ & $0.061$ && $0.453$ & $0.315$ & $0.094$ \\
				\% Succ.\ of MLE  & $1.000$ & $1.000$ & $0.000$ && $1.000$ & $1.000$ & $0.000$ \\
				\hline
				
				\multicolumn{8}{l}{\cellcolor{gray!7}\textit{Panel D. Undirected network model} $(n=100)$}\\
				Density           & $0.313$ & $0.163$ & $0.029$ && $0.342$ & $0.190$ & $0.038$ \\
				Transitivity      & $0.395$ & $0.234$ & $0.047$ && $0.453$ & $0.312$ & $0.084$ \\
				\% Succ.\ of MLE  & $1.000$ & $1.000$ & $0.000$ && $1.000$ & $1.000$ & $0.000$ \\
				\bottomrule
		\end{tabularx}}
		
		{\legend{Panel A lists parameter values used to simulate individual-specific degree heterogeneity as specified in Section \ref{sec:DGP}.
				Panels B--D report average network statistics over 1,000 Monte Carlo replications for $n=100$. 
				A.1--A.3 draw degree heterogeneity from a symmetric distribution independent of the covariates, while B.1--B.3 introduce right-skewed heterogeneity correlated with the covariates. 
				``\% Succ.\ of MLE'' denotes the fraction of replications in which the fixed effects  MLE  is successfully computed.}}
	\end{table}

	The directed covariate $X_{ij}$ is drawn from a discrete uniform distribution with possible values $\{0, 1\}$. We construct the symmetric mutual utility covariate as $Z_{ij} = \tilde{Z}_{i}\tilde{Z}_{j}$, where $\tilde{Z}_{i}$ is drawn from a discrete uniform distribution over $\{-1, 1\}$.
	The idiosyncratic error $\epsilon_{ij}$ follows a standard logistic distribution. 
	To capture the correlation between observables and unobserved degree heterogeneity, we specify the sender and receiver fixed effects as $\alpha_i = \varrho_L \mathbbm{1}_{\{\tilde{Z}_i=-1\}} + \varrho_H \mathbbm{1}_{\{\tilde{Z}_i=1\}} + \vartheta_i^{\alpha}$
	and
	$\gamma_i = \varrho_L \mathbbm{1}_{\{\tilde{Z}_i=-1\}} + \varrho_H \mathbbm{1}_{\{\tilde{Z}_i=1\}} + \vartheta_i^{\gamma}$ 
	with $\varrho_L\leq \varrho_H$. Both $\vartheta_i^{\alpha}$ and $\vartheta_i^{\gamma}$ are  drawn from the centered Beta distribution, $\text{Beta}(\varpi_0, \varpi_1) - \varpi_0/(\varpi_0 + \varpi_1)$. For the directed models, we simulate the network using a stochastic best-response dynamic, iterating 1,000 times for each dyad to ensure equilibrium convergence.\footnote{
		The directed network graph with reciprocity is simulated via an iterative updating scheme based on
		$
		g_{ij} = \mathbbm{1}{\{ X_{ij}\beta_0 + g_{ji}Z_{ij}\rho_0 +\alpha_i + \gamma_j \geq \epsilon_{ij}   \}}. 
		$
		We begin by initializing $g_{ji}=0$.  
		The directed network graph without reciprocity uses the same iterative scheme based on
		$
		g_{ij} \mathbbm{1}{\{Z_{ij}\beta_0 +\alpha_i + \gamma_j \geq \epsilon_{ij}   \}}.
		$
		The undirected network graph is generated directly by
		$
		g_{ij} = g_{ji}= \mathbbm{1}{\{ Z_{ij}\rho_0 +\alpha_i + \gamma_j \geq \epsilon_{ij}   \}}.
		$
	}
	We set the true common parameters to $\beta_0 = \rho_0 = 1$. For each design, we report results for $n \in \{100, 200\}$ based on 1,000 Monte Carlo replications.

	Table \ref{tab: Monte Carlo Designs}  summarizes six designs. 
	In Designs A.3 and B.3, the fixed effects MLE succeeds in only 0.2\%–0.3\% of replications in the reciprocal model and never succeeds in the nested models. In these same designs, the PL estimator is available in every replication and maintains near-nominal coverage for both coefficients and APEs.

\subsection{Monte Carlo Results}
Tables \ref{tab:MCresults_directedwithmutual} and \ref{tab:MCresults_directedwithmutual200} contain the Monte Carlo evidence for the reciprocal directed model. 
In the moderate dense network designs (A.1, A.2, B.1, B.2), the uncorrected  MLE and estimator based correction (EC) are available throughout. The bias of the uncorrected MLE is visible in both $\widehat{\beta}$ and $\widehat{\rho}$, and EC removes most of that bias. The penalized-likelihood estimator (PL) performs excellent in these designs.

In the sparse designs (A.3, B.3), existence becomes the primary issue. The uncorrected MLE is undefined and therefore reported as ``n.a.'' throughout the tables. The EC method  is unavailable for the same reason, since it is built on the underlying MLE. By contrast, the PL estimator is  available in every replication. This is the central finite-sample implication of Theorem \ref{thm:lambdaBC_exist}. 

The PL estimator remains accurate once sparsity becomes severe. In A.3 and B.3, the empirical coverage of PL estimates stays close to the nominal level for the common parameters. Panels C and D show the same pattern for the APEs.
In sparse networks, PL is useful not only because the estimator exists. It also supports inference on coefficients and APEs when MLE-based corrections are infeasible.
The same conclusion is even clearer when $n$ increased to 200 (Table \ref{tab:MCresults_directedwithmutual200}).

\textit{Nested Benchmark Models}.
The supplemental appendix \ref{sec:MCappendix} reports parallel results for the directed model without reciprocity and for the undirected model (Tables \ref{tab:MCresults_directedWOmutual} and  \ref{tab:MCresults_undirected}). For the directed model without reciprocity, we also report the quadruple-logit (QL) estimates of \citet{jochmans2018semiparametric}. For the undirected model, we report the tetrad-logit (TL)  benchmark of \citet{graham2017econometric}. The main pattern, however, is unchanged. As the network becomes sparse, the same existence issue arises, while the PL method remains  implementable and continues to perform well for both common parameters and APEs bias corrections. 

	\begingroup
	\fontsize{8.5pt}{9.0pt}\selectfont
	\begin{longtable}{
			>{\hspace*{1em}\raggedright\arraybackslash}p{0.16\linewidth}
			>{\centering\arraybackslash}m{0.1\linewidth}
			>{\centering\arraybackslash}m{0.1\linewidth}
			>{\centering\arraybackslash}m{0.1\linewidth}
			>{\centering\arraybackslash}m{0.025\linewidth}
			>{\centering\arraybackslash}m{0.1\linewidth}
			>{\centering\arraybackslash}m{0.1\linewidth}
			>{\centering\arraybackslash}m{0.1\linewidth}
		}
		\caption{Finite Sample Properties in  Reciprocal Directed Network Model ($n=100$)}\label{tab:MCresults_directedwithmutual}\\
		\toprule
		& \multicolumn{3}{c}{\begin{tabular}{c}
				Symmetric\\
				Uncorrelated Heterogeneity
		\end{tabular}}
		& & \multicolumn{3}{c}{\begin{tabular}{c}
				Right-Skewed\\
				Correlated Heterogeneity
		\end{tabular}}\\
		\cline{2-4}\cline{6-8}
		\rule{0pt}{2.4ex} & A.1 & A.2 & A.3 & & B.1 & B.2 & B.3 \\
		\hline
		\endfirsthead

		\hline
		\multicolumn{8}{r}{\textit{Continued on next page}} \\
		\endfoot
		
		\bottomrule
		\multicolumn{8}{l} {\legend{Panel A reports the bias of the Monte Carlo median estimates of $(\beta,\rho)$ for each  design across 1{,}000 replications; Monte Carlo standard deviation is shown in parentheses below. Panel B reports empirical coverage of nominal 95\% confidence intervals. Panel C reports the bias and standard deviation for APE estimates. Panel D reports empirical coverage of nominal 95\% confidence intervals for the APEs. ``n.a.'' indicates that the method is unavailable in the corresponding design, typically because the underlying fixed-effects MLE does not exist.}} 
		\endlastfoot

		\multicolumn{8}{l}{\cellcolor{gray!7}\textit{Panel A. Bias and standard deviation (common parameters)}}\\
		$\widehat{\beta}_{\text{MLE}}$  & $2.3$ & $2.3$ & $\mathrm{n.a.}$ & & $2.0$ & $2.1$ & $\mathrm{n.a.}$ \\
		& $(0.045)$ & $(0.054)$ &  & & $(0.045)$ & $(0.053)$ &  \\
		$\widehat{\beta}_{\text{EC}}$   & $0.2$ & $0.1$ & $\mathrm{n.a.}$ & & $0.2$ & $0.1$ & $\mathrm{n.a.}$ \\
		& $(0.045)$ & $(0.054)$ &  & & $(0.044)$ & $(0.053)$ &  \\
		$\widehat{\beta}_{\text{PL}}$   & $0.1$ & $0.0$ & $-0.3$ & & $-0.2$ & $-0.1$ & $-0.1$ \\
		& $(0.044)$ & $(0.053)$ & $(0.120)$ & & $(0.044)$ & $(0.051)$ & $(0.113)$ \\
		$\widehat{\rho}_{\text{MLE}}$ & $1.6$ & $1.1$ & $\mathrm{n.a.}$ & & $1.9$ & $1.4$ & $\mathrm{n.a.}$ \\
		& $(0.046)$ & $(0.074)$ &  & & $(0.043)$ & $(0.067)$ &  \\
		$\widehat{\rho}_{\text{EC}}$  & $0.4$ & $0.2$ & $\mathrm{n.a.}$ & & $0.4$ & $-0.1$ & $\mathrm{n.a.}$ \\
		& $(0.046)$ & $(0.075)$ &  & & $(0.043)$ & $(0.071)$ &  \\
		$\widehat{\rho}_{\text{PL}}$  & $0.3$ & $0.1$ & $1.8$ & & $0.4$ & $0.3$ & $0.9$ \\
		& $(0.046)$ & $(0.073)$ & $(0.374)$ & & $(0.042)$ & $(0.067)$ & $(0.279)$ \\
		\hline
		
		\multicolumn{8}{l}{\cellcolor{gray!7}\textit{Panel B. 95\% coverage probability (common parameters)}}\\
		$\widehat{\beta}_{\text{MLE}}$  & $0.921$ & $0.940$ & $\mathrm{n.a.}$ & & $0.933$ & $0.925$ & $\mathrm{n.a.}$ \\
		$\widehat{\beta}_{\text{EC}}$   & $0.947$ & $0.955$ & $\mathrm{n.a.}$ & & $0.962$ & $0.945$ & $\mathrm{n.a.}$ \\
		$\widehat{\beta}_{\text{PL}}$   & $0.947$ & $0.955$ & $0.953$ & & $0.967$ & $0.965$ & $0.937$ \\
		$\widehat{\rho}_{\text{MLE}}$ & $0.938$ & $0.943$ & $\mathrm{n.a.}$ & & $0.934$ & $0.937$ & $\mathrm{n.a.}$ \\
		$\widehat{\rho}_{\text{EC}}$  & $0.947$ & $0.938$ & $\mathrm{n.a.}$ & & $0.956$ & $0.941$ & $\mathrm{n.a.}$ \\
		$\widehat{\rho}_{\text{PL}}$  & $0.948$ & $0.951$ & $0.925$ & & $0.957$ & $0.939$ & $0.936$ \\
		\hline
		
		\multicolumn{8}{l}{\cellcolor{gray!7}\textit{Panel C. Bias and standard deviation (average partial effects)}}\\
		$\widehat{\beta}_{\text{MLE}}$  & $-0.1$ & $-0.5$ & $\mathrm{n.a.}$ & & $-0.1$ & $-0.1$ & $\mathrm{n.a.}$ \\
		& $(0.010)$ & $(0.009)$ &  & & $(0.009)$ & $(0.011)$ &  \\
		$\widehat{\beta}_{\text{PL}}$   & $-0.1$ & $-0.3$ & $0.5$ & & $-0.0$ & $0.0$ & $0.0$ \\
		& $(0.010)$ & $(0.009)$ & $(0.004)$ & & $(0.009)$ & $(0.011)$ & $(0.005)$ \\
		$\widehat{\rho}_{\text{MLE}}$ & $-0.0$ & $-0.4$ & $\mathrm{n.a.}$ & & $0.2$ & $-0.7$ & $\mathrm{n.a.}$ \\
		& $(0.006)$ & $(0.005)$ &  & & $(0.007)$ & $(0.008)$ &  \\
		$\widehat{\rho}_{\text{PL}}$  & $0.3$ & $0.4$ & $-1.4$ & & $0.4$ & $-0.1$ & $-2.1$ \\
		& $(0.006)$ & $(0.005)$ & $(0.001)$ & & $(0.007)$ & $(0.008)$ & $(0.002)$ \\
		\hline
		
		\multicolumn{8}{l}{\cellcolor{gray!7}\textit{Panel D. 95\% coverage probability (average partial effects)}}\\
		$\widehat{\beta}_{\text{MLE}}$  & $0.944$ & $0.959$ & $\mathrm{n.a.}$ & & $0.962$ & $0.926$ & $\mathrm{n.a.}$ \\
		$\widehat{\beta}_{\text{PL}}$   & $0.943$ & $0.960$ & $0.953$ & & $0.962$ & $0.924$ & $0.937$ \\
		$\widehat{\rho}_{\text{MLE}}$ & $0.937$ & $0.935$ & $\mathrm{n.a.}$ & & $0.915$ & $0.905$ & $\mathrm{n.a.}$ \\
		$\widehat{\rho}_{\text{PL}}$  & $0.934$ & $0.934$ & $0.925$ & & $0.916$ & $0.906$ & $0.936$ \\
	\end{longtable}
	\endgroup

	\begingroup
	\fontsize{8.5pt}{9.0pt}\selectfont
	\begin{longtable}{
			>{\hspace*{1em}\raggedright\arraybackslash}p{0.16\linewidth}
			>{\centering\arraybackslash}m{0.1\linewidth}
			>{\centering\arraybackslash}m{0.1\linewidth}
			>{\centering\arraybackslash}m{0.1\linewidth}
			>{\centering\arraybackslash}m{0.025\linewidth}
			>{\centering\arraybackslash}m{0.1\linewidth}
			>{\centering\arraybackslash}m{0.1\linewidth}
			>{\centering\arraybackslash}m{0.1\linewidth}
		}
		\caption{Finite Sample Properties in  Reciprocal Directed Network Model ($n=200$)}\label{tab:MCresults_directedwithmutual200}\\
		\toprule
		& \multicolumn{3}{c}{\begin{tabular}{c}
				Symmetric\\
				Uncorrelated Heterogeneity
		\end{tabular}}
		& & \multicolumn{3}{c}{\begin{tabular}{c}
				Right-Skewed\\
				Correlated Heterogeneity
		\end{tabular}}\\
		\cline{2-4}\cline{6-8}
		\rule{0pt}{2.4ex} & A.1 & A.2 & A.3 & & B.1 & B.2 & B.3 \\
		\hline
		\endfirsthead
		
		\hline
		\multicolumn{8}{r}{\textit{Continued on next page}} \\
		\endfoot
		
		\bottomrule
		\multicolumn{8}{l} {\legend{Same as Table \ref{tab:MCresults_directedwithmutual}.}} 
		\endlastfoot
		
		\multicolumn{8}{l}{\cellcolor{gray!7}\textit{Panel A. Bias and standard deviation (common parameters)}}\\
		$\widehat{\beta}_{\text{MLE}}$  & $1.1$ & $1.0$ & $\mathrm{n.a.}$ & & $1.2$ & $1.1$ & $\mathrm{n.a.}$ \\
		& $(0.022)$ & $(0.026)$ &  & & $(0.023)$ & $(0.027)$ &  \\
		$\widehat{\beta}_{\text{EC}}$   & $0.1$ & $-0.0$ & $\mathrm{n.a.}$ & & $0.3$ & $0.2$ & $\mathrm{n.a.}$ \\
		& $(0.021)$ & $(0.026)$ &  & & $(0.023)$ & $(0.027)$ &  \\
		$\widehat{\beta}_{\text{PL}}$   & $-0.0$ & $-0.1$ & $-0.3$ & & $0.1$ & $0.0$ & $-0.2$ \\
		& $(0.021)$ & $(0.026)$ & $(0.059)$ & & $(0.023)$ & $(0.026)$ & $(0.055)$ \\
		$\widehat{\rho}_{\text{MLE}}$ & $0.7$ & $0.5$ & $\mathrm{n.a.}$ & & $0.9$ & $0.7$ & $\mathrm{n.a.}$ \\
		& $(0.022)$ & $(0.037)$ &  & & $(0.022)$ & $(0.034)$ &  \\
		$\widehat{\rho}_{\text{EC}}$  & $0.1$ & $0.2$ & $\mathrm{n.a.}$ & & $0.1$ & $-0.1$ & $\mathrm{n.a.}$ \\
		& $(0.022)$ & $(0.038)$ &  & & $(0.021)$ & $(0.036)$ &  \\
		$\widehat{\rho}_{\text{PL}}$  & $0.1$ & $-0.0$ & $1.8$ & & $0.2$ & $0.2$ & $0.5$ \\
		& $(0.022)$ & $(0.037)$ & $(0.175)$ & & $(0.021)$ & $(0.033)$ & $(0.135)$ \\
		\hline
		
		\multicolumn{8}{l}{\cellcolor{gray!7}\textit{Panel B. 95\% coverage probability (common parameters)}}\\
		$\widehat{\beta}_{\text{MLE}}$  & $0.931$ & $0.934$ & $\mathrm{n.a.}$ & & $0.922$ & $0.926$ & $\mathrm{n.a.}$ \\
		$\widehat{\beta}_{\text{EC}}$   & $0.955$ & $0.951$ & $\mathrm{n.a.}$ & & $0.948$ & $0.938$ & $\mathrm{n.a.}$ \\
		$\widehat{\beta}_{\text{PL}}$   & $0.956$ & $0.954$ & $0.954$ & & $0.952$ & $0.935$ & $0.953$ \\
		$\widehat{\rho}_{\text{MLE}}$ & $0.941$ & $0.942$ & $\mathrm{n.a.}$ & & $0.936$ & $0.949$ & $\mathrm{n.a.}$ \\
		$\widehat{\rho}_{\text{EC}}$  & $0.955$ & $0.937$ & $\mathrm{n.a.}$ & & $0.950$ & $0.934$ & $\mathrm{n.a.}$ \\
		$\widehat{\rho}_{\text{PL}}$  & $0.958$ & $0.943$ & $0.950$ & & $0.947$ & $0.954$ & $0.939$ \\
		\hline
		
		\multicolumn{8}{l}{\cellcolor{gray!7}\textit{Panel C. Bias and standard deviation (average partial effects)}}\\
		$\widehat{\beta}_{\text{MLE}}$  & $-0.1$ & $-0.2$ & $\mathrm{n.a.}$ & & $0.1$ & $0.1$ & $\mathrm{n.a.}$ \\
		& $(0.005)$ & $(0.005)$ &  & & $(0.005)$ & $(0.006)$ &  \\
		$\widehat{\beta}_{\text{PL}}$   & $-0.0$ & $-0.2$ & $0.0$ & & $0.1$ & $0.2$ & $0.3$ \\
		& $(0.005)$ & $(0.005)$ & $(0.002)$ & & $(0.005)$ & $(0.006)$ & $(0.003)$ \\
		$\widehat{\rho}_{\text{MLE}}$ & $-0.2$ & $-0.2$ & $\mathrm{n.a.}$ & & $0.0$ & $-0.1$ & $\mathrm{n.a.}$ \\
		& $(0.004)$ & $(0.003)$ &  & & $(0.004)$ & $(0.005)$ &  \\
		$\widehat{\rho}_{\text{PL}}$  & $-0.0$ & $0.2$ & $2.7$ & & $0.1$ & $0.2$ & $0.9$ \\
		& $(0.004)$ & $(0.003)$ & $(0.001)$ & & $(0.004)$ & $(0.005)$ & $(0.001)$ \\
		\hline
		
		\multicolumn{8}{l}{\cellcolor{gray!7}\textit{Panel D. 95\% coverage probability (average partial effects)}}\\
		$\widehat{\beta}_{\text{MLE}}$  & $0.961$ & $0.953$ & $\mathrm{n.a.}$ & & $0.946$ & $0.930$ & $\mathrm{n.a.}$ \\
		$\widehat{\beta}_{\text{PL}}$   & $0.960$ & $0.952$ & $0.971$ & & $0.946$ & $0.932$ & $0.950$ \\
		$\widehat{\rho}_{\text{MLE}}$ & $0.916$ & $0.924$ & $\mathrm{n.a.}$ & & $0.908$ & $0.884$ & $\mathrm{n.a.}$ \\
		$\widehat{\rho}_{\text{PL}}$  & $0.913$ & $0.920$ & $0.954$ & & $0.910$ & $0.881$ & $0.926$ \\
	\end{longtable}
	\endgroup

	\section{Empirical Application: The Global Textiles Trade Network}\label{sec:empirical}  
	
	We apply our penalized likelihood estimator to model the formation of the 1986 global textiles trade network. Using data from the CEPII TradeProd database, we observe directed export links among 133 countries, alongside standard gravity covariates from the trade literature \citep{helpman2008estimating}, including geographic, institutional, and cultural proximity measures.\footnote{The covariates include log distance, indicators for a shared land border, island status, landlocked status, common legal origin, common language, currency union, common free trade agreement, and religious proximity. Variable definitions and descriptive statistics are reported in Appendix \ref{sec:additionalEmpiricalTable} of the supplemental material.}
	A directed link $g_{ij}=1$ indicates that country $i$ exports textiles to country $j$. 
	
	The application is useful for two reasons. 
	First, international trade is naturally directed and highly reciprocal. To consistently estimate the effects of bilateral frictions, researchers may control for unobserved degree heterogeneity via exporter and importer fixed effects.
	
	Second, the network contains boundary-degree countries even though it is not globally sparse. The standard MLE cannot accommodate such nodes and therefore requires iterative trimming.  In this application, it also removes central economies once the initial zero in-degree countries are dropped. The PL estimator avoids this sample attrition and estimates the model on the full observed network.
	
	\subsection{Network Topology and Existence Failure}
	
	Figure \ref{fig:visa} summarizes the topology of the textiles trade network. 
	The network contains 133 countries and has density 39.6 percent. It is therefore not globally sparse. 
	Reciprocity is high. Mutual dyads account for 30.1 percent of all country pairs, almost twice the rate implied by an independent network formation benchmark. 
	This high reciprocity rate validates the inclusion of the mutual utility component in our specification.
	
	Panels A and B show substantial degree heterogeneity.  
	11 countries (e.g., Albania, Bulgaria, and Vietnam) import no textiles from any partner in the network ($b_i=0$). They export textiles to at least one destination but import textiles from no country in the observed network. 
	Standard estimation protocols require dropping these 11 countries from the sample. 
	
	The trimming problem is iterative, and triggers a cascading selection problem. Removing the 11 importing countries alters the network sample for the remaining exporters. Conditional on this restricted 122 country sample, five major economies—France, Germany, Italy, Japan, and the UK—become full out-degree nodes, meaning they export to every remaining country in the data. Consequently, the standard MLE forces the deletion of these 5 global trade powerhouses.

	\begin{figure}[H]
		\centering
		\includegraphics[width=1\linewidth]{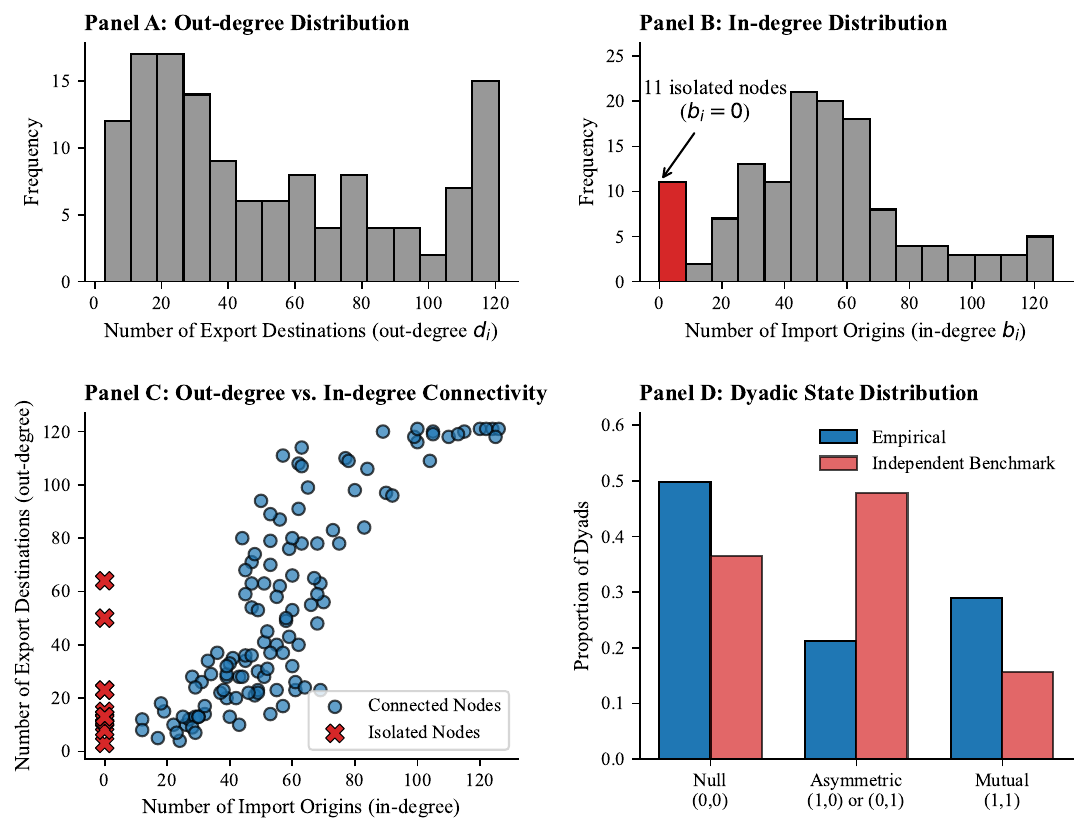}
		\caption{Degree and Dyadic State Distributions of the Textiles Trade Network. \textit{Notes}: Data are from 1986 CEPII-TradeProd for 133 countries.}\label{fig:visa}
	\end{figure}

	The final MLE sample shrinks to 117 countries. Estimating a global trade model without France, Germany, Italy, Japan, and the UK severely distorts the economic reality of the network and removes the most informative nodes. In contrast, the PL estimator retains the full 133 countries network and avoid this structural sample selection bias.

	\subsection{Estimation Results}
	Table \ref{tab:resultTrade_directed_mutual} reports the coefficient estimates and average partial effects for the reciprocal directed model.
	In Panel A, the directed utility parameters follow established gravity patterns. Distance strictly reduces the probability of a unilateral trade link (PL APE: $-0.098$), whereas sharing a common language (PL APE: $0.032$) and participating in a free trade agreement (PL APE: $0.222$) increase it.

	\begin{table}[!htbp]
		\caption{Reciprocal Directed Model: Coefficients and Average Partial Effects in the Global Textiles Trade Network}
		\label{tab:resultTrade_directed_mutual}
		\centering
		{\fontsize{8.5pt}{9.0pt}\selectfont
			\setlength{\tabcolsep}{4.2pt}%
			\renewcommand{\arraystretch}{1.05}%
			\begin{tabularx}{\linewidth}{
					>{\hspace*{1em}\raggedright\arraybackslash}X
					>{\centering\arraybackslash}p{0.14\linewidth}
					>{\centering\arraybackslash}p{0.14\linewidth}
					>{\centering\arraybackslash}p{0.04\linewidth}
					>{\centering\arraybackslash}p{0.14\linewidth}
					>{\centering\arraybackslash}p{0.14\linewidth}
				}
				\toprule
				& \multicolumn{2}{c}{Coefficients}
				& &
				\multicolumn{2}{c}{Average Partial Effects} \\
				\cline{2-3}\cline{5-6}
				\rule{0pt}{2.4ex}
				& MLE & PL & & MLE & PL \\
				\hline
				
				% ---------------- Panel A ----------------
				\multicolumn{6}{l}{\cellcolor{gray!7}\textit{Panel A: Directed Utility}}\\
				
				$\mathrm{Ln(Distance)}$   & $-1.374$\textsuperscript{***} & $-1.320$\textsuperscript{***} & & $-0.116$\textsuperscript{***} & $-0.098$\textsuperscript{***} \\
				& $(0.079)$ & $(0.070)$ & & $(0.008)$ & $(0.007)$ \\
				
				$\mathrm{Land\ border}$   & $-1.160$\textsuperscript{***} & $-0.985$\textsuperscript{***} & & $-0.103$\textsuperscript{***} & $-0.076$\textsuperscript{***} \\
				& $(0.413)$ & $(0.364)$ & & $(0.034)$ & $(0.026)$ \\
				
				$\mathrm{Island}$         & $1.274$\textsuperscript{***}  & $1.160$\textsuperscript{***}  & & $0.126$\textsuperscript{***}  & $0.099$\textsuperscript{***} \\
				& $(0.283)$ & $(0.269)$ & & $(0.029)$ & $(0.024)$ \\
				
				$\mathrm{Landlock}$       & $0.668$\textsuperscript{**}   & $0.586$\textsuperscript{*}    & & $0.065$\textsuperscript{**}   & $0.049$\textsuperscript{*} \\
				& $(0.316)$ & $(0.301)$ & & $(0.031)$ & $(0.026)$ \\
				
				$\mathrm{Legal}$          & $0.183$\textsuperscript{*}    & $0.120$                       & & $0.017$\textsuperscript{*}    & $0.010$ \\
				& $(0.101)$ & $(0.094)$ & & $(0.010)$ & $(0.008)$ \\
				
				$\mathrm{Language}$       & $0.392$\textsuperscript{***}  & $0.389$\textsuperscript{***}  & & $0.038$\textsuperscript{***}  & $0.032$\textsuperscript{***} \\
				& $(0.122)$ & $(0.115)$ & & $(0.012)$ & $(0.010)$ \\
				
				$\mathrm{Currency\ union}$& $1.972$\textsuperscript{***}  & $1.917$\textsuperscript{***}  & & $0.198$\textsuperscript{***}  & $0.167$\textsuperscript{***} \\
				& $(0.490)$ & $(0.479)$ & & $(0.051)$ & $(0.043)$ \\
				
				$\mathrm{FTA}$            & $2.494$\textsuperscript{*}    & $2.516$\textsuperscript{*}    & & $0.253$\textsuperscript{*}    & $0.222$\textsuperscript{*} \\
				& $(1.321)$ & $(1.295)$ & & $(0.135)$ & $(0.117)$ \\
				
				$\mathrm{Religion}$       & $0.635$\textsuperscript{***}  & $0.709$\textsuperscript{***}  & & $0.054$\textsuperscript{***}  & $0.052$\textsuperscript{***} \\
				& $(0.200)$ & $(0.190)$ & & $(0.017)$ & $(0.014)$ \\
				
				$\mathrm{Constant}$       & $3.918$\textsuperscript{***}  & $3.952$\textsuperscript{***}  & & $0.331$\textsuperscript{***}  & $0.292$\textsuperscript{***} \\
				& $(0.511)$ & $(0.479)$ & & $(0.045)$ & $(0.038)$ \\
				
				\hline
				
				% ---------------- Panel B ----------------
				\multicolumn{6}{l}{\cellcolor{gray!7}\textit{Panel B: Mutual Utility}}\\
				
				$\mathrm{Ln(Distance)}$   & $0.240$\textsuperscript{**}   & $0.220$\textsuperscript{**}    & & $0.017$\textsuperscript{**}   & $0.013$\textsuperscript{**} \\
				& $(0.111)$ & $(0.102)$ & & $(0.008)$ & $(0.006)$ \\
				
				$\mathrm{Land\ border}$   & $1.427$\textsuperscript{**}   & $1.249$\textsuperscript{**}   & & $0.108$\textsuperscript{**}   & $0.078$\textsuperscript{**} \\
				& $(0.615)$ & $(0.554)$ & & $(0.049)$ & $(0.035)$ \\
				
				$\mathrm{Island}$         & $-1.292$\textsuperscript{***} & $-1.148$\textsuperscript{***}  & & $-0.087$\textsuperscript{***} & $-0.067$\textsuperscript{***} \\
				& $(0.444)$ & $(0.404)$ & & $(0.028)$ & $(0.023)$ \\
				
				$\mathrm{Landlock}$       & $-0.396$                      & $-0.447$                      & & $-0.028$                      & $-0.027$ \\
				& $(0.535)$ & $(0.517)$ & & $(0.037)$ & $(0.030)$ \\
				
				$\mathrm{Legal}$          & $-0.107$                      & $-0.011$                      & & $-0.008$                      & $-0.001$ \\
				& $(0.163)$ & $(0.155)$ & & $(0.012)$ & $(0.009)$ \\
				
				$\mathrm{Language}$       & $0.097$                       & $0.140$                       & & $0.007$                       & $0.008$ \\
				& $(0.179)$ & $(0.168)$ & & $(0.013)$ & $(0.010)$ \\
				
				$\mathrm{Currency\ union}$& $0.307$                       & $0.385$                       & & $0.022$                       & $0.023$ \\
				& $(0.849)$ & $(0.840)$ & & $(0.063)$ & $(0.051)$ \\
				
				$\mathrm{FTA}$            & $0.592$                       & $0.564$                       & & $0.044$                       & $0.035$ \\
				& $(1.797)$ & $(1.784)$ & & $(0.136)$ & $(0.110)$ \\
				
				$\mathrm{Religion}$       & $0.009$                       & $-0.081$                      & & $0.001$                       & $-0.005$ \\
				& $(0.296)$ & $(0.283)$ & & $(0.021)$ & $(0.017)$ \\
				
				$\mathrm{Constant}$       & $0.140$                       & $0.087$                       & & $0.010$                       & $0.005$ \\
				& $(0.503)$ & $(0.465)$ & & $(0.036)$ & $(0.028)$ \\
				
				\hline
				
				% ---------------- Bottom block ----------------
				$\mathrm{Additional\ gravity\ controls}$ & Yes & Yes & & Yes & Yes \\
				$\mathrm{Sender/Receiver\ FEs}$ & Yes & Yes & & Yes & Yes \\
				$\mathrm{Estimated\ Sample\ (Nodes)}$ & $\mathbf{117}$ & $\mathbf{133}$ & & $\mathbf{117}$ & $\mathbf{133}$ \\
				
				\bottomrule
		\end{tabularx}}
		{\legend{This table reports the estimated coefficients and average partial effects for the directed network formation model with reciprocity, applied to the 1986 CEPII-TradeProd textiles network.  
				The full set of estimates for the nested specifications are reported in Appendix \ref{sec:additionalEmpiricalTable} of the supplemental material.
				Standard errors in parentheses. \textsuperscript{*} $p < 0.10$, \textsuperscript{**} $p < 0.05$, \textsuperscript{***} $p < 0.01$.}}
	\end{table}

	Panel B reports the mutual utility parameters, and demonstrates that bilateral variables affect unilateral and mutual trade through distinct channels. Sharing a land border strongly drives mutual trade (PL APE: $0.078$) but exhibits a negative direct effect (PL APE: $-0.076$). Conversely, island status reduces the probability of reciprocal trade (PL APE: $-0.067$) despite increasing the likelihood of unilateral links. These diverging partial effects confirm that explicitly modeling reciprocity is critical for recovering the correct mechanisms of network formation.
	
	The differences between MLE and PL are economically meaningful. Relative to the MLE, the PL estimates are systematically less pronounced, while the standard errors remain stable. For example, the directed utility APE for a land border falls in magnitude from $-0.103$ to $-0.076$, and the mutual utility APE falls from $0.108$ to $0.078$. This pattern aligns precisely with our asymptotic prediction: the incidental parameter problem primarily shifts the center of the distribution rather than scaling the variance. Failure to correct this bias, combined with the mechanical exclusion of core trading counties, may lead researchers to  overstate the  effects of trade frictions.

	\section{Concluding Remarks}\label{sec:conclusion}
	
	This paper studies fixed-effects estimation in dyadic network formation models with degree heterogeneity when the likelihood is affected by boundary-degree nodes and incidental parameter bias. 
	In empirical networks, the fixed effects MLE frequently fails to exist due to boundary-degree nodes. Even when it exists, the incidental parameter problem biases both common parameters and average partial effects. 
	We address these problems with a penalized likelihood estimator for a reciprocal directed network formation model that nests the existing undirected and nonreciprocal directed specifications. 
	The penalty keeps the estimator well defined without trimming isolated or boundary-degree nodes and yields bias corrections for both coefficients and APEs.

	The asymptotic analysis does not assume compactness of the fixed-effects parameter space. We allow the true sender and receiver effects to grow at a logarithmic rate with network size, so some link probabilities may shrink as the network expands. 
	The reciprocal model also requires a different fixed-effects Hessian approximation because sender and receiver effects for the same node are locally coupled. We derive an explicit inverse-Hessian approximation that accounts for this coupling and is sufficient for the coefficient and APE expansions.
	 
	In the empirical application, the textiles network has moderate density, yet local zero- or full-degree countries still force the standard MLE onto a selected subnetwork. The PL estimator instead uses the full observed network and yields bias-corrected coefficients and APEs.
	
	Several extensions are natural. One is likelihood-based inference built directly on the penalized objective. Another is to adapt the same logic to richer forms of dyadic dependence beyond reciprocity. These questions are left for future work. 
	  
	%-----------------------
	%  Appendix in Main text
	%-----------------------
	\begin{appendix}
		
		\section{Proofs of Main Results}\label{sec:proofmaintext}
		Appendix \ref{sec:proofmaintext} collects the proofs of the main results reported in the paper. It contains the proofs of Theorems \ref{thm:lambdaBC_exist}-- \ref{thm:lambda_asynormal}, and Corollaries \ref{thm:thetaBC_asynormal} and \ref{cor:ape_pl}.   
		The Supplementary Material (Appendices \ref{sec:clogitappendix}--\ref{sec:additionalEmpiricalTable}) collects the hexiad conditional logit construction for the reciprocal model, the proofs of Proposition \ref{thm:micro}, Lemma \ref{prop:hybrid_ranktwo_simple} and intermediate Lemmas, as well as additional Monte Carlo and empirical results.

		\begin{proof}[Proof of Theorem \ref{thm:lambdaBC_exist}] 
			Fix $\theta\in\Theta$. We first show that the penalized log-likelihood is \textit{coercive} in $\lambda$.\footnote{A function is coercive (for maximization) if it becomes arbitrarily small when parameters diverge, so the maximizer cannot escape to infinity and must lie in a bounded region.}
			Since $ D(\theta,\lambda)$ is block diagonal,  $\eta(\theta,\lambda) =     \frac12\sum_{i=1}^{n-1}\ln\det D_i(\theta,\lambda),$
			where $D_i$ is defined in \eqref{eqn:Di}.
			It therefore suffices to show that $\det D_i(\theta,\lambda)\to 0$. 
			
			Because $D_i(\theta,\lambda)$ is positive definite,
			$$
			0< \det D_i(\theta,\lambda)\le  
			\left(\textstyle\sum_{j\neq i}^{n-1}  p_{ij}(\theta,\lambda)\bigl[1-p_{ij}(\theta,\lambda)\bigr]\right)
			\left(\textstyle\sum_{j\neq i}^{n-1}p_{ji}(\theta,\lambda)\bigl[1-p_{ji}(\theta,\lambda)\bigr]\right)
			$$
			Note that
			$
			0< \sum_{j\neq i}^{n-1}  p_{ij}(\theta,\lambda)\bigl[1-p_{ij}(\theta,\lambda)\bigr] \le \frac{n-1}{4},
			$
			and
			$
			0< \sum_{j\neq i}^{n-1}  p_{ji}(\theta,\lambda)\bigl[1-p_{ji}(\theta,\lambda)\bigr] \le \frac{n-1}{4}.
			$
			In the directed network model with reciprocity,
			$$
			p_{ij}(\theta,\lambda)
			=
			\tfrac{\exp(B_{ij})+\exp(B_{ij}+B_{ji}+C_{ij})}
			{1+\exp(B_{ij})+\exp(B_{ji})+\exp(B_{ij}+B_{ji}+C_{ij})},
			$$

			Recall the identifying restriction  $\alpha_n=\gamma_n=0$. Consider first $|\alpha_i|\to\infty$. For every $j\neq i$,
			$
			B_{ij}(\theta,\lambda):=x_{ij}'\beta+\alpha_i+\gamma_j.
			$
			Since $\theta\in\Theta$ and the remaining fixed effects are bounded,
			$
			B_{ij}(\theta,\lambda)\to +\infty
			$
			if $\alpha_i\to+\infty$,
			and
			$
			B_{ij}(\theta,\lambda)\to -\infty
			$
			if $\alpha_i\to-\infty$.
			So $p_{ij}(\theta,\lambda)\to 1$ when $B_{ij}\to+\infty$, and $p_{ij}(\theta,\lambda)\to 0$ when $B_{ij}\to-\infty$. Therefore, we have
			$
			p_{ij}(\theta,\lambda)\bigl[1-p_{ij}(\theta,\lambda)\bigr]\to 0
			$
			for every $j\neq i$.
			Summing over $j\neq i$ gives
			$
			\textstyle\sum_{j\neq i}^{n-1}  p_{ij}(\theta,\lambda)\bigl[1-p_{ij}(\theta,\lambda)\bigr] \to 0.
			$
			Hence
			$
			0 < \det D_i(\theta,\lambda) \to 0.
			$
			So $\ln\det D_i(\theta,\lambda)\to -\infty$, and therefore $\eta(\theta,\lambda)\to-\infty$.
			The argument for $|\gamma_i|\to\infty$ is symmetric, which  yields $\eta(\theta,\lambda)\to-\infty$.
			
			Since the log-likelihood $\ell(\theta,\lambda)$ is bounded above by zero, whereas
			$
			\eta(\theta,\lambda)\to -\infty
			$
			whenever
			$
			\|\lambda\|_\infty\to\infty,
			$
			it follows that
			$
			\ell(\theta,\lambda)+\eta(\theta,\lambda)\to -\infty
			$
			as
			$
			\|\lambda\|_\infty\to\infty.
			$
			Thus, for each fixed $\theta$, the penalized log-likelihood is coercive in $\lambda$.

			Now let
			$
			\mathcal A_c(\theta)
			=
			\bigl\{
			\lambda\in\Lambda_n:
			\ell(\theta,\lambda)+\eta(\theta,\lambda)\ge c
			\bigr\},
			$
			where $c<\sup_{\lambda\in\Lambda_n}\{\ell(\theta,\lambda)+\eta(\theta,\lambda)\}$. By continuity of $\ell+\eta$, each upper level set $A_c(\theta)$ is closed. By coercivity, $\mathcal A_c(\theta)$ is bounded. Since $\Lambda_n$ is closed, $\mathcal A_c(\theta)$ is compact. By Weierstrass extreme value theorem,  the function $\lambda\mapsto \ell(\theta,\lambda)+\eta(\theta,\lambda)$ attains its maximum on $\Lambda_n$ for all $\theta$.\footnote{
				The coercivity is uniform in $\theta\in\Theta$, because $\Theta$ is compact and covariates have bounded support.
			} In other words, $\widehat{\lambda}_{\mathrm{PL}}$ have finite components.
			
			Finally, because $\Theta$ is compact and $\ell(\theta,\lambda)+\eta(\theta,\lambda)$ is continuous in $(\theta,\lambda)$, the same argument implies that the joint upper level set in $\Theta\times\Lambda_n$ is compact. Hence the penalized log-likelihood attains its maximum jointly in $(\theta,\lambda)$. Therefore,
			$
			(\widehat\theta_{\mathrm{PL}},\widehat\lambda_{\mathrm{PL}})
			$
			exists.
		\end{proof}

		Lemma \ref{lemma:guiyi} records  a set of rate bounds that will be used repeatedly. All bounds follow from the inverse Hessian approximation in Lemma \ref{prop:hybrid_ranktwo_simple}.  
		
		\begin{lemma}\label{lemma:guiyi}
			Suppose Assumption \ref{assumption2} and \ref{assumption:spectral_hybrid_simple} hold. Assume also that $\|\lambda_0\|=\tau \ln n$ with $\tau>0$.
			Then the following bounds hold:
			{\fontsize{10.5pt}{18pt}\selectfont
				\begin{enumerate}[label=(\roman*)]
					\item 
					$
					\bigl\|\bigl[-H_{\lambda\lambda}^{-1}(\widehat{\theta},\bar{\lambda}(\widehat{\theta}))-S(\widehat{\theta},\bar{\lambda}(\widehat{\theta}))\bigr]\partial_\lambda\ell(\widehat{\theta},\bar{\lambda}(\widehat{\theta})) \bigr\|_\infty
					=
					\mathcal{O}_P \bigl(\sqrt{\ln n} \, n^{-1+6\tau}  
					e^{6\|\bar\lambda(\widehat{\theta})-\lambda_0\|_\infty  }\bigr),
					$
					\\
					$
					\bigl\|S(\widehat{\theta},\bar{\lambda}(\widehat{\theta}))\partial_\lambda\ell(\widehat{\theta},\bar{\lambda}(\widehat{\theta})) \bigr\|_\infty
					=
					\mathcal{O}_P \bigl(\sqrt{\ln n} \, n^{-1/2+2\tau}  
					e^{2\|\bar\lambda(\widehat{\theta})-\lambda_0\|_\infty  }\bigr).
					$
					\item
					$
					\bigl\|\bigl[-H_{\lambda\lambda}^{-1}(\widehat{\theta},\bar{\lambda}(\widehat{\theta}))-S(\widehat{\theta},\bar{\lambda}(\widehat{\theta}))\bigr]H_{\lambda\theta}(\widehat{\theta},\bar{\lambda}(\widehat{\theta})) (\widehat{\theta}-\theta_0) \bigr\|_\infty
					=
					\mathcal{O}_P \bigl( n^{-1/2+6\tau}  
					e^{6\|\bar\lambda(\widehat{\theta})-\lambda_0\|_\infty  }  \|\widehat{\theta}-\theta_0\| \bigr),
					$
					\\
					$
					\bigl\|[S(\widehat{\theta},\bar{\lambda}(\widehat{\theta}))][H_{\lambda\theta}(\widehat{\theta},\bar{\lambda}(\widehat{\theta}))] (\widehat{\theta}-\theta_0) \bigr\|_\infty
					=
					\mathcal{O}_P \bigl( n^{2\tau}  
					e^{2\|\bar\lambda(\widehat{\theta})-\lambda_0\|_\infty  } \|\widehat{\theta}-\theta_0\| \bigr).
					$
					
					\item
					$
					\big\|
					-H^{-1}_{\lambda\lambda}
					\big[\sum_{k=1}^{2n-2}\partial_{\lambda\lambda'\lambda_k}\ell
					(\widehat\lambda-\lambda_0)_k\big]
					(\widehat\lambda-\lambda_0)
					\big\|_\infty
					=
					\mathcal{O}_P\bigl(n^{2\tau} \|\widehat\lambda-\lambda_0\|_\infty^2\bigr),
					$
					\\
					$
					\big\|
					-H^{-1}_{\lambda\lambda}
					\big[
					\sum_{k=1}^{2n-2}\partial_{\lambda\theta'\lambda_k}\ell(\theta_0,\widehat\lambda)
					(\widehat\lambda-\lambda_0)_k
					\big]
					(\widehat\theta-\theta_0)
					\big\|_\infty 
					=
					\mathcal{O}_P\bigl(n^{2\tau} \|\widehat\lambda-\lambda_0\|_\infty\|\widehat\theta-\theta_0\|\bigr),
					$
					\\
					$
					\big\|
					-H^{-1}_{\lambda\lambda}
					\big[
					\sum_{k=1}^{\dim\theta}\partial_{\lambda\theta'\theta_k}\ell(\widehat\theta,\lambda_0)
					(\widehat\theta-\theta_0)_k
					\big]
					(\widehat\theta-\theta_0)
					\big\|_\infty
					=
					\mathcal{O}_P\bigl(n^{2\tau}\|\widehat\theta-\theta_0\|^2\bigr).
					$
					\\
					$
					\big\|
					-H^{-1}_{\lambda\lambda}
					\big[
					\sum_{k=1}^{2n-2}\sum_{j=1}^{2n-2}
					\partial_{\lambda\lambda'\lambda_k\lambda_j}\ell(\theta_0,\widehat\lambda)
					(\widehat\lambda-\lambda_0)_k(\widehat\lambda-\lambda_0)_j
					\big]
					(\widehat\lambda-\lambda_0)
					\big\|_\infty 
					=
					\mathcal{O}_P\bigl(n^{2\tau} \|\widehat\lambda-\lambda_0\|_\infty^3\bigr).
					$
					
					\item
					$
					\bigl\|-H^{-1}_{\lambda\lambda}
					[\partial_{\lambda\lambda'}\eta(\widehat\theta,\widehat\lambda)]
					(\widehat\lambda-\lambda_0)\bigr\|_\infty
					=
					\mathcal{O}_P\bigl(n^{-1+6\tau}\|\widehat\lambda-\lambda_0\|_\infty\bigr).
					$
					
					\item 
					$
					\bigl\|-H^{-1}_{\lambda\lambda}
					[\partial_{\lambda\theta'}\eta(\widehat\theta,\widehat\lambda)]
					(\widehat\theta-\theta_0)\bigr\|_\infty
					=
					\mathcal{O}_P\bigl(n^{-1+6\tau}\|\widehat\theta-\theta_0\|\bigr).
					$
					
					\item 
					$
					\tfrac{1}{2n}
					\big[\textstyle\sum_{i=1}^{2n-2} (S \partial_{\lambda}\ell)_i\,\partial_{\theta\lambda'\lambda_i}\ell\big]
					\big(S \partial_{\lambda}\ell\big)
					=
					-\frac{1}{n}\partial_{\theta}\eta
					+
					\mathcal{O}_P(n^{2\tau-1/2}).
					$ 
					
					\item
					$
					\left\|
					-\frac12
					\big[
					\sum_{k=1}^{2n-2} (S \partial_{\lambda}\ell)_i\
					\partial_{\lambda\lambda'\lambda_i}\ell\, 
					\big]
					(S\partial_\lambda\ell)
					-
					\partial_\lambda\eta
					\right\|_\infty
					=
					\mathcal{O}_P(n^{2\tau-1/2} \ln n).
					$
				\end{enumerate}
				\par}
		\end{lemma} 
		
		To prove Theorem \ref{thm:lamba_consist}, we provide a supporting lemma to control the mean-value terms in the asymptotic expansion. 
		
		\begin{lemma}\label{lemma1_thm:lamba_consist}
			Suppose Assumption \ref{assumption2} holds and let $\theta\in\Theta$. 
			Then, for any $\mathbf{x}, \mathbf{y}$, $ \mathbf{z}$ that lie in a convex subset of $\mathbb{R}^{2 n-2}$,  
			$
			\|[H_{\lambda\lambda}(\theta, \mathbf{x})-H_{\lambda \lambda}(\theta, \mathbf{y})] \mathbf{z}\|_{\infty} \leq\|\mathbf{x}-\mathbf{y}\|_{\infty}\|\mathbf{z}\|_{\infty} 4(n-1).
			$
		\end{lemma}
		
		\begin{proof}[Proof of Theorem \ref{thm:lamba_consist}] 
			This proof contains three steps. In Step 1, we verify
			$\lVert\widehat{\lambda}(\widehat{\theta})-\bar{\lambda}(\widehat{\theta})\rVert_{\infty}=o_P(1)$ by assuming $\widehat{\theta}-\theta_0=o_P(1)$.
			When evaluating at $\theta_0$, 
			$\lVert\widehat{\lambda}(\theta_0)-\lambda_0\rVert_{\infty}=o_P(1)$ automatically holds. This feeds into Step 2 for the consistency of $\widehat{\theta}$.
			Building on previous steps, Step 3 finally gives $\lVert\widehat{\lambda}(\widehat{\theta})-\bar{\lambda}({\theta_0})\rVert_{\infty}=o_P(1)$.\footnote{Note that $\widehat{\lambda}=\widehat{\lambda}(\widehat{\theta})$ and $\bar{\lambda}({\theta_0})=\lambda_0$, by the concentrated/profile likelihood argument.}

			\#\textit{\underline{Step 1. Evaluating $\|\widehat{\lambda}(\widehat{\theta})-\bar\lambda(\widehat{\theta})\|_{\infty}$ given $\widehat{\theta}-\theta_0=o_P(1)$.} }
			
			Define the Newton iterates
			$
			\lambda^{(k+1)}(\widehat{\theta})
			=
			\lambda^{(k)}(\widehat{\theta})
			-
			H_{\lambda\lambda}^{-1}\big(\widehat{\theta},\lambda^{(k)}(\widehat{\theta})\big)\,
			\partial_\lambda \ell\big(\widehat{\theta},\lambda^{(k)}(\widehat{\theta})\big), 
			$ for $k\ge 0$.
			Define the set $\mathcal{B}_{2r_{\lambda(\widehat{\theta})}}\big(\bar\lambda(\widehat{\theta})\big)=\big\{\lambda:\ \|\lambda-\bar\lambda(\widehat{\theta})\|_{\infty}<2r_{\lambda(\widehat{\theta})}\big\}.$
			Let 
			$$
			\widetilde{\lambda}^{(k)}(\widehat{\theta}):=\lambda^{(k)}(\widehat{\theta})-\lambda^{(k-1)}(\widehat{\theta}), \text{ for any } k\ge 1.
			$$

			Let $\lambda^{(0)}(\widehat{\theta})=\bar{\lambda}(\widehat{\theta})$ and 
			$
			r_{\lambda(\widehat\theta)}
			:=
			\bigl\|
			-H_{\lambda\lambda}^{-1}\bigl(\widehat\theta,\bar\lambda(\widehat\theta)\bigr)
			\partial_\lambda\ell\bigl(\widehat\theta,\bar\lambda(\widehat\theta)\bigr)
			\bigr\|_\infty.
			$
			Then,
			$\lambda^{(1)}(\widehat{\theta})\in\mathcal{B}_{2r_{\lambda(\widehat{\theta})}}\big(\bar\lambda(\widehat{\theta})\big)$ 
			because
			$
			\|\widetilde{\lambda}^{(1)}(\widehat{\theta})\|_{\infty}
			=
			\|\lambda^{(1)}(\widehat{\theta})-\lambda^{(0)}(\widehat{\theta})\|_{\infty}
			=
			r_{\lambda(\widehat{\theta})}.
			$
			
			Assume that for any $s=1,\ldots,k$, $\lambda^{(s)}(\widehat{\theta})\in
			\mathcal{B}_{2r_{\lambda(\widehat{\theta})}}\big(\bar\lambda(\widehat{\theta})\big)$ holds.
			We will show that $\lambda^{(k+1)}(\widehat{\theta})\in
			\mathcal{B}_{2r_{\lambda(\widehat{\theta})}}\big(\bar\lambda(\widehat{\theta})\big)$.
			Consider
			\begin{align}
				\big\|\lambda^{(k+1)}(\widehat{\theta})-\lambda^{(k)}(\widehat{\theta})\big\|_{\infty}
				&=
				\big\|
				-H_{\lambda\lambda}^{-1}\big(\widehat{\theta},\lambda^{(k)}(\widehat{\theta})\big)\,
				\partial_\lambda \ell\big(\widehat{\theta},\lambda^{(k)}(\widehat{\theta})\big)
				\big\|_{\infty}.
				\label{eqn:thm2diedai}
			\end{align}
			Because $\lambda^{(k)}(\widehat\theta)$ is obtained from $\lambda^{(k-1)}(\widehat\theta)$ by one Newton step,
			\begin{align}
				[H_{\lambda\lambda}\bigl(\widehat\theta,\lambda^{(k-1)}(\widehat\theta)\bigr)]\widetilde{\lambda}^{(k)}(\widehat{\theta})
				+
				\partial_\lambda\ell\bigl(\widehat\theta,\lambda^{(k-1)}(\widehat\theta)\bigr)
				=0.
				\label{eqn:thm2onestep}
			\end{align}
			Hence
			\begin{align*}
				\partial_\lambda\ell\bigl(\widehat\theta,\lambda^{(k)}(\widehat\theta)\bigr)
				&=
				\partial_\lambda\ell\bigl(\widehat\theta,\lambda^{(k)}(\widehat\theta)\bigr)
				-
				\partial_\lambda\ell\bigl(\widehat\theta,\lambda^{(k-1)}(\widehat\theta)\bigr)
				-
				[H_{\lambda\lambda}\bigl(\widehat\theta,\lambda^{(k-1)}(\widehat\theta)\bigr)]\widetilde{\lambda}^{(k)}(\widehat{\theta}) \\
				&=
				\textstyle\int_0^1
				\bigl[
				H_{\lambda\lambda}\bigl(\widehat\theta,\lambda^{(k-1)}(\widehat\theta)+t\widetilde{\lambda}^{(k)}(\widehat{\theta})\bigr)
				-
				H_{\lambda\lambda}\bigl(\widehat\theta,\lambda^{(k-1)}(\widehat\theta)\bigr)
				\bigr]
				\widetilde{\lambda}^{(k)}(\widehat{\theta})\, \mathrm{d} t.
			\end{align*}
			By Lemma \ref{lemma1_thm:lamba_consist}, each row of the last equation is a sum of $\mathcal{O}(n)$ terms times $\|\widetilde{\lambda}^{(k)}(\widehat{\theta})\|_\infty^2$. Therefore,
			\begin{equation}
				\bigl\|
				\partial_\lambda\ell\bigl(\widehat\theta,\lambda^{(k)}(\widehat\theta)\bigr)
				\bigr\|_\infty
				=
				\mathcal{O}\bigl(n\|\widetilde{\lambda}^{(k)}(\widehat{\theta})\|_\infty^2\bigr).
				\label{eq:newton_score_bound}
			\end{equation}
			
			To apply Lemma \ref{prop:hybrid_ranktwo_simple} to \eqref{eqn:thm2diedai}, it remains to control the $u_-$-projection difference, i.e., 
			$
			u_-'  \partial_\lambda\ell\bigl(\widehat\theta,\lambda^{(k)}(\widehat\theta)\bigr)
			:=
			\sum_{i=1}^{n-1}\partial_{\alpha_i}\ell\bigl(\widehat\theta,\lambda^{(k)}(\widehat\theta)\bigr)
			-
			\sum_{i=1}^{n-1}\partial_{\gamma_i}\ell\bigl(\widehat\theta,\lambda^{(k)}(\widehat\theta)\bigr),
			$
			By calculations,\footnote{The detailed calculation of $u_-' \partial_{\lambda}\ell$ can be found in equation \eqref{eq:u_score} in the supplemental material.}
			\begin{align}
				u_-'  \partial_\lambda\ell\bigl(\widehat\theta,\lambda^{(k)}(\widehat\theta)\bigr)
				=
				C_n
				+
				\textstyle\sum_{j=1}^{n-1}
				\big(
				p_{jn}\bigl(\widehat\theta,\lambda^{(k)}(\widehat\theta)\bigr)-p_{nj}\bigl(\widehat\theta,\lambda^{(k)}(\widehat\theta)\bigr)
				\big), \label{eqn:projdiff}
			\end{align}
			where the first term
			$C_n=b_d - d_n$ does not depend on $\lambda$.
			By Assumption \ref{assumption2}, all second order derivatives of the second term in \eqref{eqn:projdiff} are uniformly bounded.
			Applying a mean value expansion around  
			$(\alpha_j^{(k-1)}(\widehat\theta),\gamma_j^{(k-1)}(\widehat\theta))$, we obtain
			$$
			u_-'  \partial_\lambda\ell\bigl(\widehat\theta,\lambda^{(k)}(\widehat\theta)\bigr)
			=
			u_-'  \partial_\lambda\ell\bigl(\widehat\theta,\lambda^{(k-1)}(\widehat\theta)\bigr)
			+
			[u_-'H_{\lambda\lambda}\bigl(\widehat\theta,\lambda^{(k-1)}(\widehat\theta)\bigr)]'[\widetilde{\lambda}^{(k)}(\widehat{\theta})]
			+
			R_k,
			$$
			where
			$
			R_k
			=
			\frac12\sum_{j=1}^{n-1}
			[\widetilde{\lambda}^{(k)}(\widehat{\theta})]_j'
			\,
			\frac{\partial^2\varphi_j(\alpha_{k,j}^*,\gamma_{k,j}^*)}
			{\partial(\alpha_j,\gamma_j)\partial(\alpha_j,\gamma_j)'}
			[\widetilde{\lambda}^{(k)}(\widehat{\theta})]_j,
			$
			with mean points $(\alpha_{k,j}^*,\lambda_{k,j}^*)$ on the line segment joining
			$(\alpha_j^{(k-1)}(\widehat\theta),\gamma_j^{(k-1)}(\widehat\theta))$ and
			$(\alpha_j^{(k)}(\widehat\theta),\gamma_j^{(k)}(\widehat\theta))$.
			The uniform bound on the second derivatives therefore gives
			$
			|R_k|
			\le
			\mathcal{O}\bigl(n\|\widetilde{\lambda}^{(k)}(\widehat{\theta})\|_\infty^2\bigr).
			$
			
			Using equation \eqref{eqn:thm2onestep} again and premultiplying by
			$u_-'$
			yields
			$$
			u_-'\partial_\lambda\ell\bigl(\widehat\theta,\lambda^{(k-1)}(\widehat\theta)\bigr)
			+
			u_-'H_{\lambda\lambda}\bigl(\widehat\theta,\lambda^{(k-1)}(\widehat\theta)\bigr)[\widetilde{\lambda}^{(k)}(\widehat{\theta})]
			=
			0.
			$$
			Therefore,
			$|u_-'H_{\lambda\lambda}\bigl(\widehat\theta,\lambda^{(k-1)}(\widehat\theta)\bigr)| =  \mathcal{O} \bigl(n\|\widetilde{\lambda}^{(k)}(\widehat{\theta})\|_\infty^2\bigr).
			$ Together with \eqref{eq:newton_score_bound} and the fact that 
			$\lambda^{(s)}(\widehat{\theta})\in\mathcal{B}_{2r_{\lambda(\widehat{\theta})}}\big(\bar\lambda(\widehat{\theta})\big)$ for any $s=1,\ldots,k$
			, applying Lemma  \ref{prop:hybrid_ranktwo_simple} gives 
			$
			\big\|   \widetilde{\lambda}^{(k+1)}(\widehat{\theta})   \big\|_{\infty}
			\leq
			\varphi(\widehat{\theta})
			\|\widetilde{\lambda}^{(k)}(\widehat{\theta})\|_\infty^2,
			$
			where $\varphi(\widehat{\theta})\le \mathcal{O}\big(n^{2\tau}  e^{4r_{\lambda(\widehat{\theta})}+2\|\bar\lambda(\widehat{\theta})-\lambda_0\|_\infty} \big)$.
			Repeating the above process generates the geometric series
			\begin{align*}
				\big\|   \widetilde{\lambda}^{(k+1)}(\widehat{\theta})   \big\|_{\infty}
				&\le
				\varphi(\widehat{\theta})^{1+2^1+\cdots+2^k}\,
				\big\|  \widetilde{\lambda}^{(1)}(\widehat{\theta})    \big\|_{\infty}^{2^{k+1}} 
				=
				\big[\varphi(\widehat{\theta})\,r_{\lambda(\widehat{\theta})}\big]^{2^{k+1}-1}\,r_{\lambda(\widehat{\theta})}.
			\end{align*}
			Thus for any $k\ge 0$,
			$
			\big\|   \widetilde{\lambda}^{(k+1)}(\widehat{\theta})  \big\|_{\infty}
			\le
			\sum_{s=0}^k \big\|   \widetilde{\lambda}^{(s+1)}(\widehat{\theta})    \big\|_{\infty} 
			=
			\sum_{s=0}^k \big[\varphi(\widehat{\theta})\,r_{\lambda(\widehat{\theta})}\big]^{\,2^s-1}\,r_{\lambda(\widehat{\theta})}.
			$
			
			By Lemma \ref{lemma:guiyi}(i),
			$
			r_{\lambda(\widehat\theta)}
			\le 
			\mathcal{O}_P \bigl(\sqrt{\ln n} \, n^{-1/2+2\tau}  
			e^{2\|\bar\lambda(\widehat{\theta})-\lambda_0\|_\infty  }\bigr),
			$
			and hence
			$$
			\varphi(\widehat{\theta})\,r_{\lambda(\widehat{\theta})}
			\le 
			\mathcal{O}_P \bigl(\sqrt{\ln n} \, n^{-1/2+4\tau}  
			e^{4r_{\lambda(\widehat{\theta})}+4\|\bar\lambda(\widehat{\theta})-\lambda_0\|_\infty  }\bigr).
			$$ 
			
			Now assume $\tau\in(0,1/8)$ and $\widehat{\theta}-\theta_0=o_P(1)$. We have 
			$r_{\lambda(\widehat\theta)}=o_P(1)$ and $\varphi(\widehat{\theta})\,r_{\lambda(\widehat{\theta})}=o_P(1)$. 
			Therefore, for large enough $n$, $\varphi(\widehat{\theta})\,r_{\lambda(\widehat{\theta})}<1/2$ with probability approaching one. 
			On this event, since $2^s-1\ge s$ for all $s\ge 0$ and $0<\varphi(\widehat{\theta})\,r_{\lambda(\widehat{\theta})}<1$, we obtain  
			$
			\big[\varphi(\widehat{\theta})\,r_{\lambda(\widehat{\theta})}\big]^{2^s-1}\le \big[\varphi(\widehat{\theta})\,r_{\lambda(\widehat{\theta})}\big]^s,
			$
			It follows that
			$
			\big\|\lambda^{(k+1)}(\widehat{\theta})-\bar{\lambda}(\widehat{\theta})\big\|_{\infty} 
			\le
			\tfrac{r_{\lambda(\widehat{\theta})}}{1-\varphi(\widehat{\theta})\,r_{\lambda(\widehat{\theta})}}
			\le 2r_{\lambda(\widehat{\theta})}.
			$. Hence,
			$
			\lambda^{(k+1)}(\widehat{\theta})
			\in
			\mathcal{B}_{2r_{\lambda(\widehat{\theta})}}
			\big(\bar{\lambda}(\widehat{\theta})\big)
			$
			for every $k\ge 0$. 
			Next, for any $m>k$,
			\begin{align*}
				\big\|\lambda^{(m)}(\widehat{\theta})-\lambda^{(k)}(\widehat{\theta})\big\|_{\infty}
				&\le
				\textstyle\sum_{s=k}^{m-1}\big\|\lambda^{(s+1)}(\widehat{\theta})-\lambda^{(s)}(\widehat{\theta})\big\|_{\infty}\\
				&\le
				\textstyle\sum_{s=k}^{\infty}\big[\varphi(\widehat{\theta})\,r_{\lambda(\widehat{\theta})}\big]^s\,r_{\lambda(\widehat{\theta})}
				=
				\tfrac{[\varphi(\widehat{\theta})\,r_{\lambda(\widehat{\theta})}]^k r_{\lambda(\widehat{\theta})}}{1-\varphi(\widehat{\theta})\,r_{\lambda(\widehat{\theta})}},
			\end{align*}
			which converges to zero as $k\to\infty$. Therefore $\{\lambda^{(k)}(\widehat{\theta})\}_{k=0}^{\infty}$ is a Cauchy sequence in $\mathbb{R}^{2n-2}$ with probability approaching one, and hence converges.
			Denote its limit by
			$
			\widehat{\lambda}(\widehat{\theta})
			:=
			\lim_{k\to\infty}\lambda^{(k)}(\widehat{\theta}).
			$
			By continuity of the iteration map, $\widehat{\lambda}(\widehat{\theta})$ is a fixed point of the iteration and thus satisfies
			$
			\partial_{\lambda}\ell\big(\widehat{\theta},\widehat{\lambda}(\widehat{\theta})\big)=0
			$
			wpa 1.
			Finally, letting $k\to\infty$ gives
			\begin{align}
				\big\|\widehat{\lambda}(\widehat{\theta})-\bar{\lambda}(\widehat{\theta})\big\|_{\infty}
				\le
				\mathcal{O}_P\big(\sqrt{\ln(n)}\,n^{2\tau-1/2}\big)=o_P(1).\label{eq:step1final}
			\end{align}

			\#\textit{\underline{Step 2. Consistency of $\widehat{\theta}$.} } 
			By manipulations, we express the log-likelihood function as
			\begin{eqnarray}
				&&\tfrac{2}{n(n-1)}\big(\ell(\theta,\lambda)-\mathbb{E}_0[\ell(\theta,\lambda)]\big) \notag\\
				&=&
				\tfrac{2}{n(n-1)}\textstyle\sum_{i}\textstyle\sum_{j\neq i}
				\big(g_{ij}g_{ji}-\mathbb{E}_0(g_{ij}g_{ji})\big)
				\big(X_{ij}'\beta+\alpha_i+\gamma_j+\tfrac{Z_{ij}'\rho}{2}\big) \notag\\
				&&+
				\tfrac{2}{n(n-1)}\textstyle\sum_{i}\textstyle\sum_{j\neq i}
				\big(g_{ij}(1-g_{ji})-\mathbb{E}_0(g_{ij}(1-g_{ji}))\big)
				\big(X_{ij}'\beta+\alpha_i+\gamma_j\big).
				\label{eqn:lminusElff}
			\end{eqnarray} 
			
			By Assumption \ref{assumption2}(ii) and (iii), $|X_{ij}'\beta|\le \kappa$ and $|Z_{ij}'\rho|\le \kappa$, where $\kappa<\infty$. If $\|\lambda\|_\infty\le \tau\ln n$ with $\tau\in(0,1/2)$, then $\|\lambda\|_\infty\le n^\tau$ for all large $n$. Hence every term in \eqref{eqn:lminusElff} are uniformly bounded by $C^{st}n^\tau$.
			Applying the triangle inequality yields
			\begin{eqnarray}
				&&\Big|\tfrac{2}{n(n-1)}\big(\ell(\theta,\lambda)-\mathbb{E}_0[\ell(\theta,\lambda)]\big)\Big| \notag\\
				&\le&
				\tfrac{3}{n}\textstyle\sum_i\Big|\tfrac{1}{n-1}\textstyle\sum_{j\neq i}
				\big(g_{ij}g_{ji}-\mathbb{E}_0(g_{ij}g_{ji})\big)\kappa\Big|
				+\tfrac{4}{n}\textstyle\sum_i\Big|\tfrac{1}{n-1}\textstyle\sum_{j\neq i}
				\big(g_{ij}g_{ji}-\mathbb{E}_0(g_{ij}g_{ji})\big)n^\tau\Big| \notag\\
				&&+
				\tfrac{2}{n}\textstyle\sum_i\Big|\tfrac{1}{n-1}\textstyle\sum_{j\neq i}
				\big(g_{ij}(1-g_{ji})-\mathbb{E}_0(g_{ij}(1-g_{ji}))\big)\kappa\Big|\notag\\
				&&+
				\tfrac{4}{n}\textstyle\sum_i\Big|\tfrac{1}{n-1}\textstyle\sum_{j\neq i}
				\big(g_{ij}(1-g_{ji})-\mathbb{E}_0(g_{ij}(1-g_{ji}))\big)n^\tau\Big|.
				\label{eqn:lminusElff2}
			\end{eqnarray} 
			For the first term, applying Hoeffding's (\citeyear{hoeffding1963}) inequality row by row gives  
			$$
			\Pr\!\Big(
			\Big|\tfrac{1}{n-1}\textstyle\sum_{j\neq i}\big(g_{ij}g_{ji}-\mathbb{E}_0(g_{ij}g_{ji})\big)\kappa\Big|
			\ge 2\kappa\sqrt{\tfrac{\ln(n-1)}{n-1}}
			\Big)
			\le 2\exp\!\big(-2\ln(n-1)\big)=\tfrac{2}{(n-1)^2}.
			$$ 
			Further, we apply Boole's inequality across $i=1,\ldots,n$ to obtain
			\begin{align*}
				&\Pr\!\Big(
				\tfrac{1}{n}\textstyle\sum_i
				\Big|\tfrac{1}{n-1}\textstyle\sum_{j\neq i}\big(g_{ij}g_{ji}-\mathbb{E}_0(g_{ij}g_{ji})\big)\kappa\Big|
				\ge 2\kappa\sqrt{\tfrac{\ln(n-1)}{n-1}}
				\Big)
				\le \tfrac{2n}{(n-1)^2}
				=\mathcal{O}(n^{-1}).
			\end{align*}
			
			As for the second term in (\ref{eqn:lminusElff2}), similarly,
			$$
			\Pr\!\Big(
			\tfrac{1}{n}\textstyle\sum_i
			\Big|\tfrac{1}{n-1}\textstyle\sum_{j\neq i}\big(g_{ij}g_{ji}-\mathbb{E}_0(g_{ij}g_{ji})\big)n^\tau\Big|
			\ge 2n^\tau\sqrt{\tfrac{\ln(n-1)}{n-1}}
			\Big)
			\le \mathcal{O}(n^{-1}).
			$$
			By the same arguments, the rest terms in (\ref{eqn:lminusElff2}) satisfies the same $ \mathcal{O}(n^{-1})$ bound.
			Substituting these bounds into \eqref{eqn:lminusElff2}, with probability at least $1- \mathcal{O}(n^{-1})$,
			\begin{equation*}
				\sup_{\theta\in\Theta,\ \|\lambda\|_\infty\le n^\tau}
				\left|
				\tfrac{2}{n(n-1)}\big(\ell(\theta,\lambda)-\mathbb{E}_0[\ell(\theta,\lambda)]\big)
				\right|
				=
				\mathcal{O}_P\big((\ln n)^{1/2}  n^{\tau-1/2}  \big).
				\label{eqn:theta_consist_pf2}
			\end{equation*}
			
			Now define 
			$
			L_n(\theta):=\tfrac{2}{n(n-1)}\ell(\theta,\widehat\lambda(\theta)),
			$
			and
			$
			Q_n(\theta):=\tfrac{2}{n(n-1)}\mathbb{E}_0[\ell(\theta,\bar\lambda(\theta))].
			$
			By Theorem \ref{thm:lamba_consist}(i), if $\tau\in(0,1/8)$,
			$\widehat\lambda(\theta_0)$ exists and 
			$
			\|\widehat\lambda(\theta_0)-\lambda_0\|_\infty
			\le
			\mathcal O_P \big(\sqrt{\ln n}\,n^{2\tau-1/2}\big),
			$ 
			so $\|\widehat{\lambda}(\theta_0)\|_\infty\le n^\tau$ with probability approaching one.  Hence 
			$
			L_n(\theta_0)
			=
			\frac{2}{n(n-1)}\mathbb E_0[\ell(\theta_0,\widehat\lambda(\theta_0))]
			+
			\mathcal{O}_P \bigl((\ln n)^{1/2}n^{\tau-1/2}\bigr).
			$
			
			Since $\bar\lambda(\theta_0)=\lambda_0$ and
			$
			\partial_\lambda \mathbb E_0[\ell(\theta_0,\lambda_0)]=0,
			$
			a second-order expansion of $\mathbb E_0[\ell(\theta_0,\lambda)]$ around $\lambda_0$ gives
			$
			\tfrac{2}{n(n-1)}
			\Bigl(
			\mathbb E_0[\ell(\theta_0,\widehat\lambda(\theta_0))]
			-
			\mathbb E_0[\ell(\theta_0,\lambda_0)]
			\Bigr)
			=
			\mathcal{O}_P  \bigl(\|\widehat\lambda(\theta_0)-\lambda_0\|_\infty^2\bigr)
			=
			o_P\bigl((\ln n)^{1/2}n^{\tau-1/2}\bigr).
			$ Therefore
			\begin{equation}
				L_n(\theta_0)
				=
				Q_n(\theta_0)
				+
				\mathcal{O}_P \bigl((\ln n)^{1/2}n^{\tau-1/2}\bigr).
				\label{eq:A8}
			\end{equation}
			Combined with the uniform consistency of $\widehat\lambda(\theta)-\bar\lambda(\theta)$ yields
			\begin{equation}
				\sup_{\theta\in\Theta}|L_n(\theta)-Q_n(\theta)|
				=
				\mathcal{O}_P\bigl((\ln n)^{1/2}n^{\tau-1/2}\bigr).
				\label{eq:A9}
			\end{equation}
			
			Define
			$
			\textsf{dis}_n(\nu)
			:=
			Q_n(\theta_0)-\sup_{\theta\in\mathcal B_\nu^c(\theta_0)}Q_n(\theta).
			$
			By the  identification of $Q_n(\theta)$ around $\theta_0$ (Assumption \ref{assumption2}(iv)), there exist constants $c>0$ and $\nu_0>0$ such that
			$
			\textsf{dis}_n(\nu)\ge c\nu^2,
			$
			and
			$  0<\nu\le \nu_0.$
			Now set
			$
			\nu_n:=M(\ln n)^{1/4}n^{\tau/2-1/4},
			$
			where $M>0$ is a sufficiently large constant. Then,
			$
			\textsf{dis}_n(\nu_n)\ge cM^2(\ln n)^{1/2}n^{\tau-1/2}.
			$
			
			Choosing $M$ large enough, \eqref{eq:A8} and \eqref{eq:A9} imply
			$
			\sup_{\theta\in\mathcal B_{\nu_n}^c(\theta_0)}L_n(\theta)
			<
			L_n(\theta_0)
			$
			with probability approaching one. Therefore,
			$
			\widehat\theta\in\mathcal B_{\nu_n}(\theta_0)
			$
			wpa 1, and thus
			$
			\|\widehat{\theta}-\theta_0\|
			=
			\mathcal{O}_P \big((\ln n)^{1/4} n^{\tau/2-1/4} \big) = o_P(1).
			$

			\#\textit{\underline{Step 3. Asymptotic Existence and Consistency of $\widehat{\lambda}$.} } 
			
			Write
			$\widehat{\lambda}-\lambda_0 = 
			\big(\widehat{\lambda}(\widehat{\theta})-\bar{\lambda}(\widehat{\theta})\big)
			+\big(\bar{\lambda}(\widehat{\theta})-\lambda_0\big).$ By the triangular inequality,
			$$
			\lVert\widehat{\lambda}-\lambda_0\rVert_{\infty}
			\le
			\lVert\widehat{\lambda}(\widehat{\theta})-\bar{\lambda}(\widehat{\theta})\rVert_{\infty}
			+\lVert\bar{\lambda}(\widehat{\theta})-\lambda_0\rVert_{\infty}.
			$$
			
			From Steps 1 and 2, we have 
			$
			\big\|\widehat{\lambda}(\widehat{\theta})-\bar{\lambda}(\widehat{\theta})\big\|_{\infty}
			\le      \mathcal{O}_P\big(\sqrt{\ln(n)}\,n^{2\tau-1/2}\big)=o_P(1),
			$
			if $\tau\in(0,1/8)$
			Therefore, it remains to evaluate $\|\bar{\lambda}(\widehat{\theta})-\lambda_0\|_\infty$. 
			The mean value theorem gives
			$
			\bar{\lambda}(\widehat{\theta})-\lambda_0
			= \partial_{\theta'}\bar{\lambda}(\theta^*)\,(\widehat{\theta}-\theta_0),
			$
			where $\theta^*$ is between $\widehat{\theta}$ and $\theta_0$. 
			Since $\bar{\lambda}(\theta)$ is the root of  
			$
			\partial_{\lambda}\mathbb{E}_0[\ell(\theta,\bar{\lambda}(\theta))]=0,
			$ differentiating $\partial_{\lambda}\mathbb{E}_0[\ell(\theta,\bar{\lambda}(\theta))]$ with respect to $\theta$
			gives
			$
			H_{\lambda\theta}(\theta,\bar{\lambda}(\theta))
			+
			\left[H_{\lambda\lambda}(\theta,\bar{\lambda}(\theta))\right]
			\,\partial_{\theta'}\bar{\lambda}(\theta)
			=0.
			$
			Thus
			\begin{equation}
				\|\bar{\lambda}(\widehat{\theta})-\lambda_0\|_\infty
				=
				\big\|
				\left[-H_{\lambda\lambda}^{-1}(\theta^*,\bar{\lambda}(\theta^*))\right]
				\left[H_{\lambda\theta}(\theta^*,\bar{\lambda}(\theta^*))\right] (\widehat{\theta}-\theta_0)\big\|_\infty.
				\label{eq:barlambda_derivative}
			\end{equation}
			Using the decomposition
			$
			-H_{\lambda\lambda}^{-1}=S+\bigl(-H_{\lambda\lambda}^{-1}-S\bigr),
			$
			together with Lemma \ref{lemma:guiyi}(ii),
			we obtain
			$
			\|\bar{\lambda}(\widehat{\theta})-\lambda_0\|_\infty
			=
			\mathcal{O}_P\big(
			n^{2\tau}
			e^{2\|\bar{\lambda}(\theta^*)-\lambda_0\|_\infty}
			\big)\|\widehat{\theta}-\theta_0\|.
			$
			In particular, if $\bar{\lambda}(\theta)$ is continuous at $\theta_0$, then by Step 2, $\|\bar{\lambda}(\theta^*)-\lambda_0\|_\infty=o_P(1)$, and hence
			$
			\|\bar{\lambda}(\widehat{\theta})-\lambda_0\|_\infty
			=
			\mathcal{O}_P\big(n^{2\tau}\big)\|\widehat{\theta}-\theta_0\|.
			$
			This completes the proof.
		\end{proof}

		Before proving the rest results,  we  establish the following intermediate lemma.

		\begin{lemma}\label{lemma2_thm:lamba_asynormal}
			Suppose Assumption \ref{assumption2} holds. Assume $\|\lambda_0\|_{\infty} \leq \tau \ln (n)$ with  $0<\tau<1/12$.
			Then,  for sufficiently large $n$
			\begin{enumerate}[label=(\roman*)]
				\item $-\tfrac{1}{N}\partial_{\theta\theta'}\ell \stackrel{p}{\rightarrow} \mathcal{I}_{\infty}+  o_P(1) $.
				\item $\tfrac{1}{\sqrt{N}}\big(\partial_\theta\ell - H_{\theta\lambda}H_{\lambda\lambda}^{-1}\partial_\lambda\ell\big)   \stackrel{d}{\rightarrow}  \mathcal N(0,\mathcal I_\infty)$.
				\item For any fixed length $L \geq 1$, 
				$\sqrt{n}[S\partial_\lambda\ell\bigr]_{1:L} \stackrel{d}{\rightarrow} \mathcal{N}\big(0, [\Omega_{\infty}]_{1: L, 1: L}\big)$.
			\end{enumerate} 
		\end{lemma}

		\begin{proof}[Proof of Theorem \ref{thm:lambda_asynormal}]
			\underline{\#\textit{Part (i)}}:
			Recall that $\widehat{\theta}:=\argmax_{\theta\in\Theta}\ell(\theta, \widehat{\lambda}(\theta))$. A mean value expansion of the first-order condition $\partial_\theta\ell(\widehat{\theta}, \widehat{\lambda}(\widehat{\theta}))=0$ at $\widehat{\theta}=\theta_0$ gives that
			$
			\sqrt{N} (\widehat{\theta}-\theta_0) = 
			-\Big[\tfrac{1}{N}\partial_{\theta\theta'}\ell(\theta^*, \widehat{\lambda}(\theta^*))\Big]^{-1}
			\tfrac{1}{\sqrt{N}}\partial_\theta\ell(\theta_0, \widehat{\lambda}(\theta_0)) 
			$ where $\theta^*$ is between $\widehat{\theta}$ and $\theta_0$.  
			Applying Lemma \ref{lemma2_thm:lamba_asynormal}(i)  gives
			$
			\sqrt{N}\,(\widehat{\theta}-\theta_0)
			=\big(\mathcal{I}_{\infty}^{-1}  +  o_P(1)  \big) \cdot
			\tfrac{1}{\sqrt{N}}  \partial_\theta\ell\big(\theta_0,\widehat{\lambda}(\theta_0)\big).
			$
			
			By a high-order expansion, 
			$
			\tfrac{1}{\sqrt{N}}\partial_\theta\ell\big(\theta_0,\widehat{\lambda}(\theta_0)\big)
			=\mathcal{T}_n^{(1)}+\mathcal{T}_n^{(2)}+\mathcal{T}_n^{(3)},
			$
			where
			\begin{align*}
				\mathcal{T}_n^{(1)}
				&=\tfrac{1}{\sqrt{N}}\big[\partial_\theta\ell+H_{\theta\lambda}\big(\widehat{\lambda}(\theta_0)-\lambda_0\big)\big],\\
				\mathcal{T}_n^{(2)}
				&=\tfrac{1}{2\sqrt{N}}
				\big[\textstyle\sum_{i=1}^{2n-2}\big[\widehat{\lambda}(\theta_0)-\lambda_0\big]_i\,\partial_{\theta\lambda'\lambda_i}\ell\big]
				\big(\widehat{\lambda}(\theta_0)-\lambda_0\big),\\
				\mathcal{T}_n^{(3)}
				&=\tfrac{1}{6\sqrt{N}}
				\big[\textstyle\sum_{i,j=1}^{2n-2}
				\big[\widehat{\lambda}(\theta_0)-\lambda_0\big]_i\big[\widehat{\lambda}(\theta_0)-\lambda_0\big]_j\,
				\partial_{\theta\lambda'\lambda_i\lambda_j}\ell(\theta_0,\breve{\lambda}(\theta_0)) \big]
				\big(\widehat{\lambda}(\theta_0)-\lambda_0\big),
			\end{align*}
			with $\breve{\lambda}(\theta_0)$ between $\widehat{\lambda}(\theta_0)$ and $\lambda_0$. 
			
			$\mathcal{T}_n^{(3)}$ is a negligible remainder term. Using Theorem \ref{thm:lamba_consist}(i) and the sparsity of  $   \partial_{\theta\lambda'\lambda_i\lambda_j}\ell $, we have
			$
			\mathcal T_n^{(3)}
			=
			\mathcal{O}_P\bigl(n\|  \widehat{\lambda}(\theta_0)-\lambda_0  \|_\infty^3\bigr)
			=
			\mathcal{O}_P\bigl((\ln n)^{3/2}n^{6\tau-1/2}\bigr).
			$

			For $\mathcal{T}_n^{(1)}$, we may replace $\widehat{\lambda}(\theta_0)-\lambda_0$ by its asymptotic representation. 
			A Taylor expansion of the first-order condition $0=\partial_\lambda \ell(\theta_0,\widehat\lambda(\theta_0))$ around $\widehat\lambda(\theta_0)=\lambda_0$ gives
			\begin{align}
				\widehat{\lambda}(\theta_0)-\lambda_0
				=
				-H_{\lambda\lambda}^{-1}  \partial_\lambda \ell 
				-\tfrac{1}{2} H_{\lambda\lambda}^{-1}  \big[\textstyle\sum_{i=1}^{2n-2}\partial_{\lambda\lambda'\lambda_i}\ell (\widetilde{\lambda}_{\theta_0})_i\big] \widetilde{\lambda}_{\theta_0} - R_{\lambda,n} 
				\label{eq:lambda_second_order_valid}
			\end{align}
			where
			$\widetilde{\lambda}_{\theta_0}:=\widehat\lambda(\theta_0)-\lambda_0$, 
			$R_{\lambda,n}=\tfrac{1}{6}H_{\lambda\lambda}^{-1}   
			\big[
			\textstyle\sum_{i,j=1}^{2n-2}
			\partial_{\lambda\lambda'\lambda_i\lambda_j}\ell(\theta_0,\lambda^*)
			(\widetilde{\lambda}_{\theta_0})_i(\widetilde{\lambda}_{\theta_0})_j
			\big]\widetilde{\lambda}_{\theta_0}$,
			with $\lambda^*$ between $\widehat{\lambda}$ and $\lambda_0$. 
			Using Lemma \ref{lemma:guiyi}(iii) and Theorem \ref{thm:lamba_consist}(i) yields
			$
			\|R_{\lambda,n}\|_\infty
			=
			\mathcal{O}_P \big(
			n^{2\tau}\|\widehat\lambda(\theta_0)-\lambda_0\|_\infty^3
			\big)
			=
			\mathcal{O}_P  \big(
			(\ln n)^{3/2} n^{8\tau-3/2}
			\big).
			$
			Using Lemma \ref{lemma:guiyi}(i) and (viii),
			replacing $\widetilde{\lambda}_{\theta_0}$ by $S \partial_\lambda\ell$ in the quadratic term in \eqref{eq:lambda_second_order_valid} yields\footnote{
				By Lemma \ref{lemma:guiyi}(i), 
				$
				\|\widetilde{\lambda}_{\theta_0}-S\partial_\lambda\ell\|_\infty
				=
				\mathcal{O}_P\bigl(\sqrt{\ln n} n^{6\tau-1}\bigr),
				$
				$
				\|\widetilde{\lambda}_{\theta_0}\|_\infty+\|S\partial_\lambda\ell\|_\infty
				=
				\mathcal{O}_P\bigl(\sqrt{\ln n}\,n^{2\tau-1/2}\bigr).
				$
				Hence
				$
				\big\|
				H_{\lambda\lambda}^{-1}
				[
				\sum_k\partial_{\lambda\lambda'\lambda_k}\ell\,
				\{(\widetilde{\lambda}_{\theta_0})_k-(S\partial_\lambda\ell)_k\}
				]\widetilde{\lambda}_{\theta_0}
				\big\|_\infty
				+
				\big\|
				H_{\lambda\lambda}^{-1}
				[
				\sum_k\partial_{\lambda\lambda'\lambda_k}\ell\,
				(S\partial_\lambda\ell)_k
				](\widetilde{\lambda}_{\theta_0}-S\partial_\lambda\ell)
				\big\|_\infty
				\leq 
				\mathcal{O}_P\big( 
				n^{2\tau}
				\|\widetilde{\lambda}_{\theta_0}-S\partial_\lambda\ell\|_\infty
				\bigl(\|\widetilde{\lambda}_{\theta_0}\|_\infty+\|S\partial_\lambda\ell\|_\infty\bigr)\big)=\mathcal{O}_P\big( (\ln n)^{3/2}n^{10\tau-3/2}   \big).
				$
			}
			\begin{align}
				\widehat{\lambda}(\theta_0)-\lambda_0
				&=
				-H_{\lambda\lambda}^{-1}\partial_\lambda\ell
				-\tfrac12
				H_{\lambda\lambda}^{-1}
				\big[
				\textstyle\sum_{k=1}^{2n-2}
				\partial_{\lambda\lambda'\lambda_k}\ell\,
				(S\partial_\lambda\ell)_k
				\big](S\partial_\lambda\ell)
				+\widetilde R_{\lambda,n}, \notag \\
				&=
				-H_{\lambda\lambda}^{-1}\partial_\lambda\ell
				+ H_{\lambda\lambda}^{-1}  \partial_{\lambda}\eta
				+\widetilde R_{\lambda,n}, 
				\label{eq:lambda_second_order_with_S}
			\end{align}
			with
			$
			\|\widetilde R_{\lambda,n}\|_\infty
			=
			\mathcal{O}_P\big( (\ln n)^{3/2}n^{10\tau-3/2}  \big).
			$ Since $\|H_{\theta\lambda}\|_{\infty} = \mathcal{O}(n^2)$, we obtain 
			\begin{align}
				\mathcal{T}_n^{(1)} &=\tfrac{1}{\sqrt{N}}\big(\partial_\theta\ell - H_{\theta\lambda} H_{\lambda\lambda}^{-1}\partial_\lambda\ell\big)
				+ \tfrac{1}{\sqrt{N}}H_{\theta\lambda}H_{\lambda\lambda}^{-1}\partial_{\lambda}\eta 
				+ \mathcal{O}_P\big( (\ln n)^{3/2}n^{10\tau-1/2}   \big).
			\end{align}
			
			It remains to evaluate $\mathcal{T}_n^{(2)}$. Since $\|\widetilde{\lambda}_{\theta_0}-S\partial_\lambda\ell\|_\infty = \mathcal{O}_P\bigl(\sqrt{\ln n} n^{6\tau-1}\bigr)$, we have
			$
			\mathcal{T}_n^{(2)}
			=\tfrac{1}{2\sqrt{N}}
			\big[\textstyle\sum_{i=1}^{2n-2}(S\partial_{\lambda}\ell)_i\,\partial_{\theta\lambda'\lambda_i}\ell\big] (S\partial_{\lambda}\ell)
			+
			\mathcal{O}_P\big(  \ln n \, n^{12\tau-1}    \big).
			$
			Applying Lemma \ref{lemma:guiyi}(vi),
			\begin{align}
				\mathcal{T}_n^{(2)}
				=-\tfrac{1}{2\sqrt{N}}\partial_{\theta}\eta
				+
				\mathcal{O}_P\big(  \ln n \, n^{12\tau-1}    \big)+
				\mathcal{O}_P\big( n^{2\tau-1/2}    \big). \label{eqn:A14}
			\end{align}
			
			Combing foregoing results, we have
			\begin{align}
				\sqrt{N}\,(\widehat{\theta}-\theta_0)
				=&
				\mathcal{I}_{\infty}^{-1}  \tfrac{1}{\sqrt{N}}\big(\partial_\theta\ell - H_{\theta\lambda} H_{\lambda\lambda}^{-1}\partial_\lambda\ell\big) 
				-
				\mathcal{I}_{\infty}^{-1}  \tfrac{1}{\sqrt{N}} \big(\partial_{\theta}\eta - H_{\theta\lambda}H_{\lambda\lambda}^{-1}\partial_{\lambda}\eta \big) \notag \\
				&+ \mathcal{O}_P\big( (\ln n)^{3/2}n^{10\tau-1/2}  \big) +\mathcal{O}_P\big(  \ln n \, n^{12\tau-1}   \big)+ \mathcal{O}_P\big( n^{2\tau-1/2}   \big) \label{eq:A15}
			\end{align}
			Assuming $\tau\in(0,1/12)$, above $\mathcal{O}_P$ terms are all $o_P(1)$. The limiting distribution of $\sqrt{N}\,(\widehat{\theta}-\theta_0)$ follows from the application of Lemma \ref{lemma2_thm:lamba_asynormal}(ii) and Slutsky's theorem.

			\underline{\#\textit{Part (ii)}}:
			By the first-order condition for $\widehat\lambda$, i.e.,
			$0=\partial_\lambda \ell(\widehat\theta,\widehat\lambda).$
			Using the decomposition
			$
			-H_{\lambda\lambda}^{-1}=S+\bigl(-H_{\lambda\lambda}^{-1}-S\bigr)
			$, and mean value expansion around $(\theta_0,\lambda_0)$ yields
			\begin{eqnarray}
				\widehat{\lambda}-\lambda_0
				&=&
				S \partial_\lambda \ell  
				+
				\underbrace{
					S H_{\lambda\theta}(\widehat{\theta}-\theta_0)}_{\|.\|_{\infty} = \mathcal{O}_P(n^{2\tau-1})}
				+  
				\underbrace{
					(-H_{\lambda\lambda}^{-1}-S)  \partial_\lambda \ell  }_{\|.\|_{\infty} = \mathcal{O}_P(\sqrt{\ln n} n^{6\tau-1})}
				+ 
				\underbrace{
					(-H_{\lambda\lambda}^{-1}-S)  H_{\lambda\theta}(\widehat{\theta}-\theta_0) }_{\|.\|_{\infty} = \mathcal{O}_P(n^{6\tau-3/2})} \notag \\
				&&
				\underbrace{
					-\tfrac{1}{2}H_{\lambda\lambda}^{-1}  \big[\textstyle\sum_{k=1}^{2n-2}\partial_{\lambda\lambda'\lambda_k}\ell  (\widehat\lambda-\lambda_0)_k\big] (\widehat\lambda-\lambda_0)}_{
					= H_{\lambda\lambda}^{-1}\partial_{\lambda}\eta + \widehat{R}_{\lambda,n} \text{ by \eqref{eq:lambda_second_order_with_S} from part(i)}
				} \notag \\
				&& 
				\underbrace{
					-H_{\lambda\lambda}^{-1}   
					\big[
					\textstyle\sum_{k=1}^{2n-2}\partial_{\lambda\theta'\lambda_k}\ell(\theta_0,\lambda^*)
					(\widehat\lambda-\lambda_0)_k
					\big]
					(\widehat\theta-\theta_0)    }_{\|.\|_{\infty} = \mathcal{O}_P((\ln n)^{1/2} n^{4\tau-3/2})} \notag \\
				&& 
				\underbrace{
					-H_{\lambda\lambda}^{-1}    
					\big[
					\textstyle\sum_{k=1}^{\dim\theta}\partial_{\lambda\theta'\theta_k}\ell(\theta^*,\lambda_0)
					(\widehat\theta-\theta_0)_k
					\big]
					(\widehat\theta-\theta_0)   }_{\|.\|_{\infty} = \mathcal{O}_P(n^{2\tau-2})} \notag \\
				&& 
				\underbrace{
					-\tfrac{1}{6}H_{\lambda\lambda}^{-1}   
					\big[
					\textstyle\sum_{k=1}^{2n-2}\sum_{j=1}^{2n-2}
					\partial_{\lambda\lambda'\lambda_k\lambda_j}\ell(\theta_0,\check\lambda)
					(\widehat\lambda-\lambda_0)_k(\widehat\lambda-\lambda_0)_j
					\big]
					(\widehat\lambda-\lambda_0)}_{\|.\|_{\infty} = \mathcal{O}_P((\ln n)^{3/2} n^{6\tau-3/2})}
				,
				\label{eq:lambda_foc_full} 
			\end{eqnarray}
			where  $\theta^*$ is between $\widehat{\theta}$ and $\theta_0$,
			and $\lambda^*$ and $\check\lambda$ are between $\widehat{\lambda}$ and $\lambda_0$.
			
			Note that $ \|\widehat{\theta}-\theta_0\|=\mathcal{O}_P(n^{-1})$ from part (i), and  $\|\widehat{\lambda}-\lambda_0\|_{\infty}=\mathcal{O}_P(\sqrt{\ln n} n^{2\tau-1/2})$ from Theorem \ref{thm:lamba_consist}(iii). 
			By Lemma \ref{lemma:guiyi}(iii),  terms  with underbrace are all bound by $o_P(n^{-1/2})$ in the sense of $\ell_{\infty}$ norm, if $\tau\in(0,1/12)$.
			Equivalently,
			$\| \widehat{\lambda}-\lambda_0 - S\partial_{\lambda}\ell\|_{\infty} = o_P(n^{-1/2})$.
			Finally, by Lemma \ref{lemma2_thm:lamba_asynormal}(iii),
			$\sqrt{n}(\widehat{\lambda}-\lambda_0)_{1:L} = \sqrt{n}(S \partial_\lambda \ell)_{1:L} + o_P(1) \stackrel{d}{\rightarrow} \mathcal{N}\big(0, [\Omega_{\infty}]_{1: L, 1: L}\big),$ for any fixed integer $L\geq 1$.
		\end{proof}

		\begin{proof}[Proof of Corollary \ref{thm:thetaBC_asynormal}]
			The proof uses the same arguments as Theorems \ref{thm:lamba_consist} and \ref{thm:lambda_asynormal}. For the expansion as Theorem \ref{thm:lambda_asynormal}, the only difference is that the concentrated score contains the additional term $\partial_\theta\eta$ in ${\mathcal T}_n^{(2)} $.
			We show only the different steps. 
			
			First, the consistency argument is unchanged. 
			$\|\partial_\lambda\eta\|_\infty$ are of smaller order than the score. $\|\partial_\lambda\eta\|_\infty=\mathcal{O}(1)$, whereas
			$\|\partial_\lambda\ell\|_\infty=\mathcal{O}_P(\sqrt{n\log n})$. The Newton argument for  $\widehat\lambda(\theta)$ therefore applies to $\widehat\lambda_{\mathrm{PL}}(\theta)$. 
			On the $n$ dependent neighborhood used in the Step 2 of the proof of Theorem \ref{thm:lamba_consist},
			$
			\sup_{\theta,\lambda}
			\big|
			\frac{2}{n(n-1)}\eta(\theta,\lambda)
			\big|
			=
			\mathcal{O} \left({\ln n} \, {n^{-1}}\right)
			=
			o(1).
			$
			So the penalty term does not affect the consistency of $\widehat\theta_{\mathrm{PL}}$, i.e., $\widehat\theta_{\mathrm{PL}}  \stackrel{p}{\rightarrow} \theta_0$ holds.
			
			Now use the first-order condition for the concentrated penalized objective function 
			$
			0
			=
			\partial_\theta\ell(\widehat\theta_{\mathrm{PL}},
			\widehat\lambda_{\mathrm{PL}}(\widehat\theta_{\mathrm{PL}}))
			+
			\partial_\theta\eta(\widehat\theta_{\mathrm{PL}},
			\widehat\lambda_{\mathrm{PL}}(\widehat\theta_{\mathrm{PL}})).
			$
			A mean-value expansion around $\theta_0$ gives
			$$
			\sqrt N(\widehat\theta_{\mathrm{PL}}-\theta_0)
			=
			\big[
			-\tfrac{1}{N}
			\partial_{\theta\theta'}
			\ell_{\mathrm{PL}}(\theta^*,
			\widehat\lambda_{\mathrm{PL}}(\theta^*))
			\big]^{-1}
			\tfrac{1}{\sqrt N}
			\big[
			\partial_\theta\ell(\theta_0,\widehat\lambda_{\mathrm{PL}}(\theta_0))
			+
			\partial_\theta\eta(\theta_0,\widehat\lambda_{\mathrm{PL}}(\theta_0))
			\big],
			$$
			where $\theta^*$ lies between $\widehat\theta_{\mathrm{PL}}$ and $\theta_0$.
			The Hessian contribution of the penalty is negligible. Indeed,
			$\partial_{\theta\theta'}\eta=\mathcal{O}(n)$, and the mixed terms generated by profiling out
			$\lambda$ are also $\mathcal{O}(n)$. Hence by  Lemma \ref{lemma2_thm:lamba_asynormal}(i),
			$
			-\frac{1}{N}
			\partial_{\theta\theta'}
			\ell_{\mathrm{PL}}(\theta^*,
			\widehat\lambda_{\mathrm{PL}}(\theta^*)) 
			=
			\mathcal I_\infty+o_P(1).
			$
			Therefore
			$$
			\sqrt N(\widehat\theta_{\mathrm{PL}}-\theta_0)
			=
			\mathcal I_\infty^{-1}
			\tfrac{1}{\sqrt N}
			\big[
			\partial_\theta\ell(\theta_0,\widehat\lambda_{\mathrm{PL}}(\theta_0))
			+
			\partial_\theta\eta(\theta_0,\widehat\lambda_{\mathrm{PL}}(\theta_0))
			\big]
			+o_P(1).
			\label{eq:PL_theta_expansion_core}
			$$

			We next expand the score terms  in the bracket.  
			For the penalty score, by $
			\|\widehat\lambda_{\mathrm{PL}}(\theta_0)-\lambda_0\|_\infty=o_P(1),
			$ we have
			$
			\frac{1}{\sqrt N}
			\partial_\theta\eta(\theta_0,\widehat\lambda_{\mathrm{PL}}(\theta_0))
			=
			\frac{1}{\sqrt N}\partial_\theta\eta
			+o_P(1),
			$
			and note $\frac{1}{\sqrt N}\partial_\theta\eta = \mathcal{O}(1)$.
			Similar to the part (i) of the proof of Theorem \ref{thm:lambda_asynormal}, a Taylor expansion of the score terms   around $\lambda_0$ gives
			$
			\frac{1}{\sqrt N}
			\partial_\theta\ell(\theta_0,\widehat\lambda_{\mathrm{PL}}(\theta_0))
			+ \frac{1}{\sqrt N}\partial_\theta\eta(\theta_0,\widehat\lambda_{\mathrm{PL}}(\theta_0))
			=
			\widetilde{\mathcal T}_n^{(1)}
			+
			\widetilde{\mathcal T}_n^{(2)}
			+
			\widetilde{\mathcal T}_n^{(3)}.
			$

			For $\widetilde{\mathcal T}_n^{(1)}
			=
			\frac{1}{\sqrt N}
			\big[
			\partial_\theta\ell
			+
			H_{\theta\lambda}
			(\widehat\lambda_{\mathrm{PL}}(\theta_0)-\lambda_0) 
			\big] $, similar to \eqref{eq:lambda_second_order_with_S}, the asymptotic representation of $\widehat\lambda_{\mathrm{PL}}(\theta_0)-\lambda_0$ becomes
			\begin{align}
				\widehat\lambda_{\mathrm{PL}}(\theta_0)-\lambda_0 
				&=
				-H_{\lambda\lambda}^{-1}\partial_\lambda\ell
				+ H_{\lambda\lambda}^{-1}  (\partial_{\lambda}\eta-\partial_{\lambda}\eta)
				+\widetilde R_{\lambda,n}
				=
				-H_{\lambda\lambda}^{-1}\partial_\lambda\ell 
				+\widetilde R_{\lambda,n}, 
				\label{eq:lambda_second_order_with_S22}
			\end{align}
			where $\widetilde R_{\lambda,n}=o_P(n^{-1/2})$ whenever  $\tau\in(0,1/12)$.
			Hence, by Lemma \ref{lemma2_thm:lamba_asynormal}(ii), we obtain that
			$
			\widetilde{\mathcal T}_n^{(1)}
			=
			\frac{1}{\sqrt N}
			\left[
			\partial_\theta\ell
			-
			H_{\theta\lambda}H_{\lambda\lambda}^{-1}\partial_\lambda\ell
			\right]
			+o_P(1) \stackrel{d}{\rightarrow}
			\mathcal N(0,\mathcal I_\infty).
			$
			
			The quadratic term $\widetilde{\mathcal T}_n^{(2)}$ is the same leading term as in Theorem \ref{thm:lambda_asynormal}. The PL analogous version to \eqref{eqn:A14} is 
			$
			\widetilde{\mathcal T}_n^{(2)} = \tfrac{1}{2\sqrt{N}}(-\partial_{\theta}\eta+\partial_{\theta}\eta) +o_P(1)  = o_P(1)
			$
			if $\tau\in(0,1/12)$. 
			The penalty score exactly offset the leading bias in $\widetilde{\mathcal T}_n^{(2)}$.
			The cubic remainder is again negligible and satisfies $\widetilde{\mathcal T}_n^{(3)}=o_P(1)$ .  
			
			Combining the foregoing results, the limiting distribution follows: 
			$
			\sqrt N(\widehat\theta_{\mathrm{PL}}-\theta_0)
			\stackrel{d}{\rightarrow}
			\mathcal N(0,\mathcal I_\infty^{-1}).
			$
			This proves the corollary. 
		\end{proof}

		\begin{proof}[Proof of Corollary \ref{cor:ape_pl}] 
			Let $\widetilde{\theta}_{\mathrm{PL}} :=\widehat\theta_{\mathrm{PL}}-\theta_0$
			and
			$\widetilde{\lambda}_{\mathrm{PL}} :=\widehat\lambda_{\mathrm{PL}}-\lambda_0$.
			As a bookkeeping from the proof of Corollary \ref{thm:thetaBC_asynormal}, \eqref{eq:lambda_foc_full} in part (ii) of the proof of  Theorem \ref{thm:lambda_asynormal} implies  when $\tau\in(0,1/12)$, 
			$\|\widetilde{\lambda}_{\mathrm{PL}}\|\leq n^{1/2}\| \widetilde{\lambda}_{\mathrm{PL}}  \|_{\infty} = \mathcal{O}_P(n^{-1/2})$,
			$\| \widetilde{\lambda}_{\mathrm{PL}} - S\partial_{\lambda}\ell \| \leq n^{1/2} \| \widetilde{\lambda}_{\mathrm{PL}} - S\partial_{\lambda}\ell \|_{\infty} = o_P(1)$
			, and 
			$\| \widetilde{\lambda}_{\mathrm{PL}} - (-H_{\lambda\lambda}^{-1})\partial_\lambda \ell - S H_{\lambda\theta} (\widetilde{\theta}_{\mathrm{PL}}) \|\leq n^{1/2}\| \widetilde{\lambda}_{\mathrm{PL}} - (-H_{\lambda\lambda}^{-1}) \partial_\lambda \ell  - S H_{\lambda\theta} (\widetilde{\theta}_{\mathrm{PL}}) \|_{\infty} = o_P(n^{-1/2})$.
			
			A Taylor expansion of 
			$\Delta(\widehat\theta_{\mathrm{PL}},\widehat\lambda_{\mathrm{PL}})$ around $(\theta_0, \lambda_0)$ yields 
			\begin{align}
				\Delta(\widehat\theta_{\mathrm{PL}},\widehat\lambda_{\mathrm{PL}}) - \Delta = 
				[\partial_{\theta'}\Delta] \widetilde{\theta}_{\mathrm{PL}}
				+ [ \partial_{\lambda'}\Delta ] \widetilde{\lambda}_{\mathrm{PL}}
				+ \tfrac{1}{2}\widetilde{\lambda}_{\mathrm{PL}}'   [\partial_{\lambda\lambda'}\Delta]  \widetilde{\lambda}_{\mathrm{PL}}
				+ R_{\Delta,n} \label{eqn:Deltaexpand}
			\end{align}
			where 
			\begin{eqnarray*}
				R_{\Delta,n} &=&
				\tfrac{1}{2}\widetilde{\theta}_{\mathrm{PL}}'  [\partial_{\theta\theta'}\Delta(\beta^*,\lambda_0)]   \widetilde{\theta}_{\mathrm{PL}}
				+ \widetilde{\theta}_{\mathrm{PL}}'   [\partial_{\theta\lambda'}\Delta(\beta_0,\lambda^*)]        \widetilde{\lambda}_{\mathrm{PL}} \\
				&&+ \tfrac{1}{6} \textstyle\sum_{i=1}^{2n-2} \widetilde{\lambda}_{\mathrm{PL}}' [\partial_{\lambda\lambda'\lambda_i}\Delta(\beta_0,\check\lambda)] \widetilde{\lambda}_{\mathrm{PL}}(\widetilde{\lambda}_{\mathrm{PL}})_i,
			\end{eqnarray*}
			with $\theta^*$ between $\widehat\theta_{\mathrm{PL}}$ and $\theta_0$, and  $\lambda^*$ and $\check\lambda$ between $\widehat\lambda_{\mathrm{PL}}$ and $\lambda_0$.
			
			We verify the remainder $R_{\Delta,n}$. Note that from Corollary \ref{thm:thetaBC_asynormal}, we have $\|\widetilde{\theta}_{\mathrm{PL}}\|=\mathcal{O}_P(n^{-1})$, 
			and $\|\widetilde{\lambda}_{\mathrm{PL}}\|=\mathcal{O}_P(1)$. 
			Since $\|\partial_{\theta\theta'}\Delta\|=\mathcal{O}(1)$, the first term in $R_{\Delta,n}$ is $\mathcal{O}_P(n^{-2})$. 
			Because $\|\partial_{\theta\lambda'}\Delta\|\ =\mathcal{O}(n^{-1/2})$, the second mixed term is $\mathcal{O}_P(n^{-3/2})$.
			The third order term in  $\lambda$ is again $\mathcal{O}_P(n^{-3/2})$.\footnote{
				Since  $\Delta_{ij}$ depends on fixed effects only through the two  blocks
				$(\lambda_i,\lambda_j)$, dyadic level sparsity of $\partial_{\lambda\lambda\lambda}\Delta$ gives
				$
				\big|
				\sum_{i}
				\widetilde{\lambda}_{\mathrm{PL}}'
				[\partial_{\lambda\lambda'\lambda_i}\Delta(\beta_0,\check\lambda)]
				\widetilde{\lambda}_{\mathrm{PL}}
				(\widetilde{\lambda}_{\mathrm{PL}})_i
				\big|
				\le
				\frac{C}{n(n-1)}
				\sum_{i\neq j}
				(\|\widetilde{\lambda}_{\mathrm{PL},i}\|+\|\widetilde{\lambda}_{\mathrm{PL},j}\|)^3 .
				$
				Using
				$
				(x+y)^3\le C\|\widetilde{\lambda}_{\mathrm{PL}}\|_\infty(x^2+y^2),
				$
				we obtain
				$
				\left|
				\sum_{i}
				\widetilde{\lambda}_{\mathrm{PL}}'
				\bigl[\partial_{\lambda\lambda'\lambda_r}\Delta(\beta_0,\check\lambda)\bigr]
				\widetilde{\lambda}_{\mathrm{PL}}
				(\widetilde{\lambda}_{\mathrm{PL}})_i
				\right|
				\le
				\mathcal{O}_P \big(n^{-1}   \|\widetilde{\lambda}_{\mathrm{PL}}\|_\infty\|\widetilde{\lambda}_{\mathrm{PL}}\|^2
				\big)
				= \mathcal{O}_P(n^{-3/2}),
				$
				where we used  
				$\|\widetilde{\lambda}_{\mathrm{PL}}\|_\infty=\mathcal{O}_P(n^{-1/2}),$ and $\|\widetilde{\lambda}_{\mathrm{PL}}\|=n^{1/2} \|\widetilde{\lambda}_{\mathrm{PL}}\|_\infty= \mathcal{O}_P(1)$.
			}
			Thus the remainder term satisfies   $  |R_{\Delta,n}| \leq o_P(n^{-1}).$
			
			We next evaluate the linear terms in \eqref{eqn:Deltaexpand}. 
			Using  
			$\| \widetilde{\lambda}_{\mathrm{PL}} - (-H_{\lambda\lambda}^{-1})\partial_\lambda \ell - S H_{\lambda\theta} (\widetilde{\theta}_{\mathrm{PL}}) \|= o_P(n^{-1/2})$, 
			and
			$ \widetilde{\theta}_{\mathrm{PL}}  = \mathcal{I}_{\infty}^{-1}  \tfrac{1}{{N}}\big(\partial_\theta\ell - H_{\theta\lambda} H_{\lambda\lambda}^{-1}\partial_\lambda\ell\big) + o_P(n^{-1})$,
			we have
			\begin{align}
				\sqrt{N}\big(    [\partial_{\theta'}\Delta] \widetilde{\theta}_{\mathrm{PL}} + [ \partial_{\lambda'}\Delta ] \widetilde{\lambda}_{\mathrm{PL}}    \big)
				=U_{\Delta}  + o_P(1), \label{eq:ape_linear_part}
			\end{align}
			where $U_{\Delta}:=  ([\partial_{\theta'}\Delta ]+ [\partial_{\lambda'}\Delta]  S H_{\lambda\theta} ) \mathcal{I}_{\infty}^{-1}  \tfrac{1}{{\sqrt{N}}}\big(\partial_\theta\ell - H_{\theta\lambda} H_{\lambda\lambda}^{-1}\partial_\lambda\ell\big)  +   \sqrt{N} [\partial_{\lambda'}\Delta ]  (-H_{\lambda\lambda}^{-1}) \partial_\lambda \ell.$
			
			We now treat the quadratic term in \eqref{eqn:Deltaexpand}. 
			Since  
			$\| \widetilde{\lambda}_{\mathrm{PL}} - S \partial_\lambda \ell  \|= o_P(1)$ and 
			$
			\|\partial_{\lambda\lambda'}\Delta\|=\mathcal{O}(n^{-1}),
			$ 
			we have
			$
			\widetilde{\lambda}_{\mathrm{PL}}'   [\partial_{\lambda\lambda'}\Delta]   \widetilde{\lambda}_{\mathrm{PL}}
			= (S \partial_\lambda \ell)' [\partial_{\lambda\lambda'}\Delta]  (S \partial_\lambda \ell)  
			+o_P(n^{-1}).
			$
			By the information identify $\mathbb{E}_0(\partial_{\lambda}\ell\partial_{\lambda'}\ell)= -H_{\lambda\lambda} $ and Lemma \ref{lemma:guiyi}(i),
			$$
			\mathbb{E}_0\big[  S \partial_\lambda \ell)' [\partial_{\lambda\lambda'}\Delta]  (S \partial_\lambda \ell)   \big]
			= 
			\operatorname{tr}\big([\partial_{\lambda\lambda'}\Delta] S(-H_{\lambda\lambda}) S \big)
			= 
			\operatorname{tr}\big([\partial_{\lambda\lambda'}\Delta] S \big) + o(n^{-1}).
			$$
			Therefore, 
			$
			\sqrt N \big(\tfrac{1}{2} \widetilde{\lambda}_{\mathrm{PL}}'   [\partial_{\lambda\lambda'}\Delta]   \widetilde{\lambda}_{\mathrm{PL}}  - \tfrac{1}{2}\operatorname{tr}\big([\partial_{\lambda\lambda'}\Delta] S \big)   \big) =   o_P(n^{-1}). 
			$
			Together with \eqref{eqn:Deltaexpand} and \eqref{eq:ape_linear_part}, 
			$$
			\sqrt N
			\big(
			\Delta(\widehat\theta_{\mathrm{PL}},\widehat\lambda_{\mathrm{PL}})
			-
			\Delta(\theta_0,\lambda_0)
			-
			\tfrac12\operatorname{tr}(\partial_{\lambda\lambda'}\Delta\,S)
			\big)
			=
			U_{\Delta}+o_P(1).
			$$
			
			To complete the proof, it remains to analyze the asymptotic distribution of $U_{\Delta}$.
			Note that $U_{\Delta}$ is a linear combination of mean zero scores. By Lemma \ref{lemma2_thm:lamba_asynormal}(ii) and (iii), the asymptotic distribution of $U_{\Delta}$ follows. For the asymptotic variance, we note that two score components (i.e., $\tfrac{1}{{\sqrt{N}}}\big(\partial_\theta\ell - H_{\theta\lambda} H_{\lambda\lambda}^{-1}\partial_\lambda\ell\big)  $ and $ \partial_{\lambda}\ell$) in $U_{\Delta}$ are  orthogonal:
			$
			\mathbb{E}_0\big[  \tfrac{1}{{\sqrt{N}}}\big(\partial_\theta\ell - H_{\theta\lambda} H_{\lambda\lambda}^{-1}\partial_\lambda\ell\big)  \partial_{\lambda'}\ell   \big] = 0. 
			$
			Thus, the asymptotic variance has the form of
			\begin{align}
				\mathrm{Var}_0(U_{\Delta}) = \mathbb{E}_0(U_{\Delta}U_{\Delta}')
				=& ([\partial_{\theta'}\Delta ]+ [\partial_{\lambda'}\Delta]  S H_{\lambda\theta} )
				\mathcal{I}_{\infty}^{-1}
				([\partial_{\theta}\Delta ]+   H_{\theta\lambda}S [\partial_{\lambda}\Delta]) \notag \\
				&+ N   [\partial_{\lambda'}\Delta ]  S  [\partial_{\lambda}\Delta ]
				=: V_{\Delta,n}  \label{eq:APEvariance}
			\end{align}
			Thus, 
			$
			\sqrt N
			\big(
			\Delta(\widehat\theta_{\mathrm{PL}},\widehat\lambda_{\mathrm{PL}})
			-
			\Delta(\theta_0,\lambda_0)
			-
			\tfrac12\operatorname{tr}(\partial_{\lambda\lambda'}\Delta\,S)
			\big)
			=
			U_{\Delta}+o_P(1)  \overset{d}{\rightarrow} \mathcal N(0, V_{\Delta,\infty}).
			$
			This completes the proof.
		\end{proof}

	\end{appendix}

	%%%%%%%%%%%%%%%%%%%%%%%%%%%%%%%%%%%%%%%%%%%%%%
	%% Bibliography:                            %%
	%%%%%%%%%%%%%%%%%%%%%%%%%%%%%%%%%%%%%%%%%%%%%%
	%% IMPORTANT: References in the bibliography should be complete, 
	%% including the first and last names, and date of publication.
	
	%% If your bibliography is in bibtex format, uncomment commands:
	%\bibliographystyle{qe} % Style BST file for qe
	\bibliographystyle{ecta-fullname} % Style BST file for ecta
	\bibliography{reference}  % Bibliography file (usually '*.bib')

\end{document}